\documentclass[a4paper,11pt]{article}

\usepackage{hyperref}
\usepackage{graphicx,fullpage,amsmath,enumerate,amssymb,bbm,url,amsthm,multirow,caption,subcaption,amsfonts}

\newtheorem{thm}{Theorem}[section]
\newtheorem{cor}[thm]{Corollary}
\newtheorem{lem}[thm]{Lemma}
\newtheorem{prop}[thm]{Proposition}
\theoremstyle{definition}
\newtheorem{defn}[thm]{Definition}
\theoremstyle{remark}
\newtheorem{rem}[thm]{Remark}

\newtheorem{alg}{Algorithm}

\newcommand{\reals}{\mathbb{R}}
\newcommand{\complex}{\mathbb{C}}

\newcommand{\tr}{\mathsf{tr}}
\newcommand{\dete}{\mathsf{det}}

\newcommand{\var}{\mathsf{var}}
\newcommand{\tran}{^{\top}}

\newcommand{\pP}{\mathbb{P}}
\newcommand{\pE}{\mathbb{E}}

\newcommand{\im}{\mathsf{Im}}
\newcommand{\re}{\mathsf{Re}}
\newcommand{\argu}{\mathsf{arg}}
\newcommand{\Argu}{\mathsf{Arg}}
\newcommand{\diag}{\mathsf{diag}}

\begin{document}

\title{Exact Simulation of Wishart Multidimensional Stochastic Volatility Model}
\date{\today}
\author{
    Chulmin Kang\thanks{National Institute for Mathematical Sciences, Republic of Korea, E-mail: ckang@nims.re.kr}
    \and
    Wanmo Kang\thanks{Department of Mathematical Sciences, KAIST, Republic of Korea, E-mail: wanmo.kang@kaist.edu}
 } \maketitle

\begin{abstract}
In this article, we propose an exact simulation method of the Wishart multidimensional stochastic volatility (WMSV) model, which was recently introduced
by Da Fonseca et al. \cite{DGT08}. Our method is based on
analysis of the conditional characteristic function of the log-price given volatility level. In particular, we found an explicit expression for the
conditional characteristic function for the Heston model. We perform numerical experiments to demonstrate the performance and accuracy of our method.
As a result of numerical experiments, it is shown that our new method is much faster and reliable than Euler discretization method.
\end{abstract}


\section{Introduction}
Ever since the Heston's stochastic volatility model \cite{He93} was introduced in $1993$, it has been the most important and widely used model among
stochastic volatility models. Its popularity relies on the clear financial interpretation of parameters and computational tractability of the model.
Heston \cite{He93} found the characteristic function of logarithmic asset price in a closed-form and showed that European call options can be priced by
inverting the characteristic function.

Despite its popularity, extensive empirical research has documented limitations of the Heston model
\cite{Ba96,Ba00,CW07,CHJ09,DG11}. The most critical deficiency of the model is that it can not generate realistic term structure of
volatility smiles; Heston model provides too flat implied volatility surface to capture reality. But empirical studies revealed that the implied
volatility curve for a short maturity has rather steep slope and it is convex, and that for a long maturity tends to be linear
\cite{CHJ09,DG11}. Consequently, much effort has been focused on generalizing the Heston model so as to accommodate such stylized facts.
There are two approaches to generalize the Heston model. The first one is to add jumps in the dynamics of stock return, the volatility, or both of
them \cite{Ba96,Ba00,DPS00}. And the other stream has been investigating the multifactor nature of the implied volatility
\cite{Ba00,CHJ09,DGT08,DPS00}.

Among multifactor stochastic volatility models, the Wishart multidimensional stochastic volatility (WMSV) model is the most flexible one,
and it matches the term structure of implied volatilities well \cite{DGT08,DG11}. Term structure of the realized volatilities
in this model is described by a positive semidefinite matrix valued stochastic process, namely the Wishart
process which was developed by Bru \cite{Br91} and introduced in the financial literature by Gourieroux and Sufana \cite{GS10}.
The dependence between the asset price and the volatility factor is also characterized by a square matrix. Such a matrix specification of the
model make it possible to capture the stylized facts observed in the option markets. The model can fit both the long-term volatility level
and the short-term skew at the same time. In addition, the model exhibits the stochastic leverage effect, and it makes the model adequate to deal with a stochastic skew effect.
Even with these flexible parameter specifications, it is still analytically tractable; the problem of pricing European vanilla options can be handled through the
transform methods of Duffie et al. \cite{DPS00} and the FFT methods of Carr and Madan \cite{CM99}.

Even though analytic aspects of the WMSV model are well explored \cite{BBE08,DGT08,DG11,GP11}, there are only few studies on the
simulation of the model. Gauthier and Possama\"{i} \cite{GP12} proposed some discretization schemes beyond the crude Euler scheme, but their
schemes are less satisfactory in terms of accuracy. Their schemes entails severe bias errors, and its
accuracy is sensitive to the model parameters.

In this paper, we propose an exact simulation method of the WMSV model, which does not suffer from discretization bias. Our new simulation method
for the WMSV model is motivated by the exact simulation method of Broadie and Kaya \cite{BK06} for the Heston model. Their key observation was that
the Heston model can be sampled exactly, provided that the endpoints and the integral of variance process are sampled exactly. For this purpose, they
derived the conditional characteristic function for the integral of the variance process up to the endpoint conditional on its endpoints, and they used it to sample the
integral by Fourier inversion techniques. We take a similar, but rather direct approach to generate a sample from the WMSV model.
We first sample the terminal value for the volatility factor process. Then we use the Fourier inversion techniques to generate the log-price
of the stock conditional on the given terminal value of the volatility factor.

In order to apply the Fourier inversion techniques, it is necessary to find the relevant characteristic function. In our case, it is the conditional
characteristic function of the log-price given the terminal value of volatility factor. We prove that it can be obtained by solving a certain system
of ordinary differential equations. In particular, we provide an explicit formula of the conditional characteristic function for the Heston model.
In general, the system of ordinary differential equations does not admit a closed-form solution due to non-commutativity of matrix multiplications,
but it can be efficiently solved by numerical methods.

The rest of this paper is organized as follows. Section \ref{sec:theory} introduces the WMSV model specifications and derive the conditional
characteristic function of the log-price. Section \ref{sec:BK} provides a brief review on the Broadie-Kaya's exact simulation method of the Heston
model. Section \ref{sec:exact} presents our exact simulation method for the WMSV model in detail. In Section \ref{sec:numer}, we give some numerical
results and compare our method with a standard discretization method and Broadie-Kaya's method. Section \ref{sec:con} conclude the paper.
Some detailed derivations are deferred to appendices.

\section{The Wishart Multidimensional Stochastic Volatility Model}\label{sec:theory}

Within the WMSV model, the dynamics of asset price, under the risk neutral measure, is described by
\begin{equation} \label{eq:asset}
    \frac{dS_t}{S_t} = r dt + \tr\Big[\sqrt{X_t}\big(dW_t R^{\top} + dZ_t \sqrt{I_d - RR^{\top}}\big)\Big],
        \hspace{0.5cm} S_0 = s > 0,
\end{equation}
where $\tr$ is the trace operator, $r$ is a constant which represents the risk neutral drift, $W$ and $Z$ are independent $d \times d$ standard matrix Brownian motions,
i.e., all the entries of $W$ and $Z$ are independent standard $1$-dimensional Brownian motions. In this specification, the volatility is determined by the
$d\times d$ symmetric positive semidefinite matrix-valued process $X_t$. In the followings, we denote by $S_d^{++} (S_d^+)$ the set of symmetric positive (semi)definite matrices.
The $d \times d$ matrix $R$ specifies the correlation between
asset price and volatility factors, and it determines the skewness of the distribution of the return.

The volatility factor process is assumed to be a Wishart process which solves the equation
\begin{equation}\label{eq:wishart}
    dX_t = (\delta \Sigma^{\top} \Sigma + H X_t + X_t H^{\top})dt + \sqrt{X_t}dW_t \Sigma + \Sigma^{\top} dW_t^{\top} \sqrt{X_t},
        \hspace{0.5cm} X_0 = x \in S_d^+,
\end{equation}
where $\Sigma$, $H$ are $d \times d$ matrices, and $\delta \ge d - 1$.
The parameter restriction $\delta \ge d - 1$ ensures the existence of the unique weak solution of (\ref{eq:wishart}) \cite{CFMT11}.
The Wishart process is a matrix analog of the square-root mean-reverting process. In order to grant the typical mean-reverting feature of the volatility,
the matrix $H$ is assumed to be negative definite.

Throughout the paper, we express the asset price
in terms of the log-price $Y_t = \log(S_t)$, so that (\ref{eq:asset}) becomes
\begin{equation}\label{eq:logprice}
    dY_t = \big(r - \frac{1}{2}\tr[X_t]\big) dt + \tr\Big[\sqrt{X_t}\big(dW_t R^{\top} + dZ_t \sqrt{I_d - RR^{\top}}\big)\Big],
        \hspace{0.5cm} Y_0 = y.
\end{equation}
We first review the affine transform formula for the WMSV model \cite{DGT08}, and we derive the conditional Laplace transform of
log-price $Y_T$ given the terminal volatility $X_T$ using the affine transform formula and the change of measure techniques.

\subsection{Laplace Transform of $Y_T$}
Due to the affine nature of the WMSV model, the Laplace transform of the log-price process is exponentially affine in the initial values $(x, y)$.
In particular, the Laplace transform is of the following form.

\begin{prop}[Da Fonseca et al. \cite{DGT08}]
The Laplace transform of the log-price $Y_T$ is given by
\begin{equation}\label{eq:Laplace}
    \pE_{x,y}\Big[ e^{-u Y_T} \Big] = e^{- \phi(0,u) - \tr[\psi(0,u) x] - u y}, \hspace{0.5cm} \text{ if LHS is finite},
\end{equation}
where $(\phi, \psi)$ is the solution of
\begin{equation}\label{eq:RDE}
    \left\{
        \begin{array}{ll}
            \partial _t \psi(t, u) =& 2 \psi(t, u) \Sigma^{\top} \Sigma \psi(t, u)\\
                &- (H^{\top} - u R\Sigma)\psi(t,u) - \psi(t,u)(H - u\Sigma^{\top} R^{\top}) + \frac{u(u+1)}{2} I_d,\\
            \partial_t \phi(t, u) =& -\delta \tr[\psi(t,u)\Sigma^{\top}\Sigma] - ur,
        \end{array}
    \right.
\end{equation}
for $0 \le t \le T$, with the terminal value $\psi(T, u) = 0$ and $\phi(T, u) = 0$.
\end{prop}

\begin{rem}
In this paper, we take the equations (\ref{eq:RDE}) as backward equations for notational convenience. With this backward equations,
the conditional Laplace transform given $\mathcal{F}_t$ can be written as
\begin{equation*}
    \pE_{x,y}\Big[e^{- u Y_T}\Big|\mathcal{F}_t\Big] = \pE_{\scriptscriptstyle X_t, Y_t} \Big[e^{- u Y_{T-t}}\Big]
        = e^{- \tilde{\phi}(0,u) - \tr[\tilde{\psi}(0,u) X_t] - u Y_t},
\end{equation*}
where $(\tilde{\phi}, \tilde{\psi})$ is the solution of (\ref{eq:RDE}) with $\tilde{\psi}(T-t, u) = 0$ and $\tilde{\phi}(T-t, u) = 0$.
Notice that $\psi(\cdot + t, u) = \tilde{\psi}(\cdot, u)$ and $\phi(\cdot + t, u) = \tilde{\phi}(\cdot, u)$ by the uniqueness of the solution. So we have
$\tilde{\psi}(0, u)= \psi(t, u)$ and $\tilde{\phi}(0) = \phi(t)$. Hence we have
\begin{equation}\label{eq:cATS}
    \pE_{x,y}\Big[e^{- u Y_T}\Big|\mathcal{F}_t\Big] = e^{- \phi(t,u) - \tr[\psi(t,u) X_t] - u Y_t}.
\end{equation}
\end{rem}

\subsection{Conditional Laplace Transform of $Y_T$ given $X_T$}\label{sec:condLaplace}
The aim of this subsection is to find the conditional Laplace transform of $Y_T$ given $X_T = x_{\scriptscriptstyle T} \in S_d^{++}$. We apply the affine
transform formula (\ref{eq:cATS}) and the change of measure technique (e.g. see Theorem XI.3.2 of \cite{RY99})
to find the conditional Laplace transform. From this section, we assume that
$\Sigma$ is nonsingular, and $\delta > d-1$ to ensure the nonsingularity of Wishart process $X_t$.

To state and prove the main result of this section, it is necessary to recall the definitions of the noncentral Wishart distribution and some multivariate special functions.
We collect them in Appendix \ref{sec:def}.

\begin{thm}\label{thm:main}
The conditional Laplace transform of log-price $Y_T$ given $X_T = x_{\scriptscriptstyle T} \in S_d^{++}$ satisfies
\begin{eqnarray}
    \lefteqn{ \pE_{x,y} \Big[e^{- u Y_T}\Big|X_T = x_{\scriptscriptstyle T} \Big] = \left(\frac{\dete[V(0,0)]}{\dete[V(0,u)]}\right)^{\delta/2}
        \exp\Big\{\textstyle  - \phi(0,u) - u y  \Big\}}\nonumber\\
    & & \times \exp\Big\{\textstyle -\frac{1}{2}
        \tr\big[ (2 \psi(0, u) + \Psi(0, u) V(0,u)^{-1} \Psi(0, u)^{\top} - \Psi(0,0) V(0,0)^{-1} \Psi(0,0)^{\top}) x\big]\Big\}\nonumber\\
    & & \times \exp\Big\{\textstyle  -\frac{1}{2} \tr\big[ (V(0,u)^{-1} - V(0,0)^{-1})x_{\scriptscriptstyle T} \big]\Big\} \label{eq:condChf}\\
    & & \times \frac{ {}_0 F_1\Big(\textstyle \frac{1}{2}\delta ;
        \frac{1}{4} V(0,u)^{-1}\Psi(0,u)^{\top} x \Psi(0,u) V(0,u)^{-1} x_{\scriptscriptstyle T}\Big)}
        {{}_0 F_1\Big(\textstyle \frac{1}{2}\delta ;
        \frac{1}{4} V(0,0)^{-1} \Psi(0,0)^{\top} x \Psi(0,0) V(0,0)^{-1} x_{\scriptscriptstyle T}\Big) },\nonumber
\end{eqnarray}
where ${}_0 F_1$ is the hypergeometric function of matrix argument defined in Appendix \ref{sec:def}, the matrix-valued functions $\psi$, $\Psi$, $V$, and the real-valued function $\phi$ are the solution of the system of ordinary
differential equations:
\begin{equation}\label{eq:equations}
    \left\{
        \begin{array}{rcl}
            \partial _t \psi(t, u) &=& 2 \psi(t, u) \Sigma^{\top} \Sigma \psi(t, u)\\
            & &  - (H^{\top} - u R\Sigma)\psi(t,u) - \psi(t,u)(H - u\Sigma^{\top} R^{\top}) + \frac{u(u+1)}{2} I_d,\\
            \partial_t \phi(t,u) &=& -\delta \tr[\psi(s,u) \Sigma^{\top}\Sigma] - u r,\\
            \partial_t \Psi(t,u) &=& - (H^{\top} - u R \Sigma  - 2\psi(t,u) \Sigma^{\top}\Sigma) \Psi(t,u),\\
            \partial_t V(t,u) &=& - \Psi(t,u)^{\top} \Sigma^{\top}\Sigma \Psi(t,u),\\
        \end{array}
    \right.
\end{equation}
for $0 \le t \le T$, with terminal values $\psi(T, u) = V(T, u) = 0$, $\Psi(T, u) = I_d$, and $\phi(T, u) = 0$.
\end{thm}

\begin{proof}
Using the affine transform formula (\ref{eq:cATS}), we define a positive martingale
\begin{eqnarray*}
    Z_t &=& \pE_{x,y}\Big[ \exp\big\{ - u Y_T + \phi(0,u) + \tr[\psi(0,u)x] + u y\big\}\Big | \mathcal{F}_t\Big]\\
        &=& \exp\big\{ - \phi(t,u) + \phi(0,u) -\tr[\psi(t,u)X_t] + \tr[\psi(0,u)x] - u(Y_t - y)\big\},
\end{eqnarray*}
for $0 \le t \le T$, where $(\phi, \psi)$ is the solution of the equations (\ref{eq:RDE}). Since $Z_0 = 1$ and $Z_t > 0$ for
all $0 \le t \le T$, it can be used as a Radon-Nikodym derivative process. We define an equivalent measure $\tilde{\pP}$ by
\begin{equation*}
    \frac{d \tilde{\pP}}{d \pP} = Z_T, \hspace{0.7cm} \text{ on } \hspace{0.5cm} \mathcal{F}_T.
\end{equation*}
In order to apply Girsanov theorem, we need to find a martingale $(M_t)_{\scriptscriptstyle 0 \le t \le T}$ such that
$Z_t = \mathcal{E}(M)_t$. We apply the integration by parts formula to have
\begin{eqnarray*}
    \lefteqn{\tr[\psi(t,u) X_t] - \tr[\psi(0,u) x] = \int_0^t \tr[\partial_s \psi(s,u) X_s] ds + \int_0^t \tr[\psi(s,u) dX_s]}\\
        &=& \int_0^t  \tr\big[(\partial_s \psi(s,u) + H^{\top} \psi(s,u) + \psi(s,u) H)X_s\big] ds\\
        & & + \int_0^t \delta \tr[\psi(s,u)\Sigma^{\top}\Sigma] ds + \int_0^t \tr\big[ 2\Sigma \psi(s,u) \sqrt{X_s} dW_s\big],
\end{eqnarray*}
and
\begin{equation*}
    u(Y_t - y) = \int^t_0 \big( u r - \frac{1}{2} \tr[u X_s] \big) ds + \int_0^t \tr\big[ u R^{\top}\sqrt{X_s} dW_s\big]
        + \int_0^t \tr\big[ u \sqrt{I_d - RR^{\top}}\sqrt{X_s} dZ_s\big].
\end{equation*}
Since $(\phi, \psi)$ satisfies the equations (\ref{eq:RDE}),
\begin{eqnarray*}
    \lefteqn{-\phi(t,u) + \phi(0,u) - \tr[\psi(t,u) X_t] + \tr[\psi(0,u) x] - u (Y_t - y)}\\
    &=& - \int_0^t \tr\big[(2\Sigma\psi(s,u) + u R^{\top})\sqrt{X_s}dW_s\big] - \int_0^t \tr\big[ u \sqrt{I_d - RR^{\top}}\sqrt{X_s} dZ_s\big]\\
    & & - \frac{1}{2} \int_0^t \tr\big[ (4\psi(s,u)\Sigma^{\top}\Sigma \psi(s,u)
        + 2uR\Sigma\psi(s,u) + 2u\psi(s,u) \Sigma^{\top}R^{\top} + u^2 I_d) X_s\big] ds .
\end{eqnarray*}
Set
\begin{equation*}
    M = - \int_0^{\cdot} \tr\big[(2\Sigma\psi(s,u) + u R^{\top})\sqrt{X_s}dW_s\big]
        - \int_0^{\cdot} \tr\big[ u \sqrt{I_d - RR^{\top}}\sqrt{X_s} dZ_s \big].
\end{equation*}
Then
\begin{equation*}
    \langle M \rangle_t =  \int_0^t \tr\big[ (4\psi(s,u)\Sigma^{\top}\Sigma \psi(s,u)
        + 2uR\Sigma\psi(s,u) + 2u\psi(s,u) \Sigma^{\top}R^{\top} + u^2 I_d) X_s\big] ds.
\end{equation*}
Therefore, we have $Z_t = \mathcal{E}(M)_t$, $0 \le t \le T$. By Girsanov
theorem, the dynamics of $X$ is as follows
\begin{equation*}
    dX_t = (\delta\Sigma^{\top} + H(t,u) X_t + X_t H(t,u)^{\top})dt + \sqrt{X_t}d\tilde{W}_t\Sigma + \Sigma^{\top}d\tilde{W}_t^{\top} \sqrt{X_t},
\end{equation*}
where $H(t,u) = H - u \Sigma^{\top} R^{\top} - 2\Sigma^{\top}\Sigma \psi(t,u)$, and $\tilde{W}$ is a $d \times d$ matrix Brownian motion under
$\tilde{\pP}$. Therefore $X$ is a Wishart process with time-varying linear drift in
the sense of Appendix \ref{sec:wisharttime}.

According to Proposition \ref{prop:wisharttimelaplace} in Appendix \ref{sec:wisharttime}, under $\tilde{\pP}$, $X_T$ has the noncentral Wishart distribution
$\mathcal{W}_{d}(\delta, V(0, u), V(0, u)^{-1} \Psi(0, u)^{\top} x \Psi(0,u))$,
and it has the transition density from time $0$ to $T$ :
\begin{eqnarray}
    p_{\scriptscriptstyle 0, T}(x, x_{\scriptscriptstyle T};u) &=& \frac{(\dete [x_{\scriptscriptstyle T}])^{(\delta - d - 1)/2}}
        {2^{d\delta/2}\Gamma_d(\textstyle \frac{1}{2} \delta) (\dete [V(0,u)])^{\delta/2}}
        \exp\Big\{\textstyle - \frac{1}{2} \tr\big[ V(0,u)^{-1}(x_{\scriptscriptstyle T} + \Psi(0,u)^{\top} x \Psi(0,u))\big]\Big\}
        \nonumber\\
    & & \times {}_0 F_1\Big(\textstyle \frac{1}{2}\delta ;
        \frac{1}{4} V(0,u)^{-1}\Psi(0,u)^{\top} x \Psi(0,u) V(0,u)^{-1} x_{\scriptscriptstyle T}\Big),\label{eq:qden}
\end{eqnarray}
where $\Psi(t,u)$ and $V(t,u)$ solve the corresponding equations in (\ref{eq:equations}).
Notice that the transition density of $X$ under $\pP$ can be obtained by taking $u = 0$ because $H(t,0) = H$ for all $0 \le t \le T$.

Now we are ready to compute the conditional Laplace transform of $Y_T$ given $X_T = x_{\scriptscriptstyle T}$. Observe that
\begin{eqnarray*}
    \lefteqn{\int_{S_d^{++}} \pE_{x,y} \big[ e^{-u Y_T} \big| X_T = x_{\scriptscriptstyle T} \big]
        f(x_{\scriptscriptstyle T}) p_{\scriptscriptstyle 0, T}(x, x_{\scriptscriptstyle T};0) dx_{\scriptscriptstyle T}
        = \pE_{x,y} \Big[ e^{-u Y_T} f(X_T) \Big]}\\
    & & = e^{ - \phi(0,u) - \tr[ \psi(0,u) x] - u y} \pE_{x,y} \Big[ Z_T f(X_T)\Big]
        = e^{ - \phi(0,u) - \tr[ \psi(0,u) x] - u y} \tilde{\pE}_{x,y}\Big[ f(X_T)\Big]\\
    & & = e^{ - \phi(0,u) - \tr[ \psi(0,u) x] - u y}
        \int_{S_d^{++}} f(x_{\scriptscriptstyle T}) p_{\scriptscriptstyle 0, T}(x, x_{\scriptscriptstyle T};u) dx_{\scriptscriptstyle T}.
\end{eqnarray*}
for all nonnegative measurable function $f$ on $S_d^{++}$. Hence the conditional Laplace transform satisfies
\begin{equation*}
    \pE_{x,y} \big[ e^{-u Y_T} \big| X_T = x_{\scriptscriptstyle T} \big]
        = e^{ - \phi(0,u) - \tr[ \psi(0,u) x] - u y}
        \frac{p_{\scriptscriptstyle 0, T}(x, x_{\scriptscriptstyle T};u)}{p_{\scriptscriptstyle 0, T}(x, x_{\scriptscriptstyle T};0)}.
\end{equation*}
By substituting (\ref{eq:qden}) into the above equation, we complete the proof.
\end{proof}

\subsection{Conditional Laplace Transform for Heston Model}\label{sec:LapHes}
If the volatility has a single factor, the WMSV model reduces to the classical Heston's stochastic volatility
model \cite{He93}. Therefore the analysis in the previous subsection can be readily applied to the Heston model. Moreover,
(\ref{eq:equations}) admits a closed-form solution, and it opens possibility of further analysis of the conditional Laplace transform
for the Heston model.

In the Heston model, the dynamics of log-price process $Y$ and variance
process $X$, under the risk neutral measure, are described by the
following system of stochastic differential equations:
\begin{equation}\label{eq:hest}
    \left\{
        \begin{array}{lll}
            dX_t = & \kappa(\theta - X_t)dt + \sigma\sqrt{X_t} dW_t, \\
            dY_t = & \big(r - \frac{1}{2}X_t\big) dt + \sqrt{X_t}\big[ \rho dW_t + \sqrt{1-\rho^2} dZ_t\big],
        \end{array}
    \right.
\end{equation}
with initial value $X_0 = x \ge 0$ and $Y_0 = y \in \reals$.
The second equation gives the dynamics of the log-price process: the stock price $S_t$ at time $t$ is given by $S_t = e^{Y_t}$, $r$ is the
risk-neutral drift and $\sqrt{X_t}$ is the volatility. The first equation describes the dynamics of the stochastic variance process. The parameter
$\kappa > 0$ determines the speed of the mean reversion, $\theta > 0$ represents the long-run mean of the variance process, and $\sigma > 0$ is the
volatility of the variance process. $W$ and $Z$ are independent standard $1$-dimensional Brownian motions.

The parameters of the WMSV model and those of the Heston model are related in the following way:
\begin{equation*}
    \kappa\theta = \delta\Sigma^{\top}\Sigma = \delta \Sigma^2, \hspace{0.5cm}
    \sigma = 2\Sigma, \hspace{0.5cm} \kappa = -2 H, \hspace{0.5cm} \rho = R.
\end{equation*}
And (\ref{eq:equations}) are reduced to
\begin{equation}\label{eq:systemheston}
    \left\{
        \begin{array}{rcl}
            \partial_t \psi(t,u) &=& \frac{1}{2} \sigma^2 \psi(t, u) + (\kappa + u \sigma \rho) \psi(t,u) + \frac{1}{2} u(u+1),\\
            \partial_t \phi(t,u) &=& -\kappa \theta \psi(t,u) - u r,\\
            \partial_t \Psi(t,u) &=&  \frac{1}{2} (\kappa + u \sigma \rho + \sigma^2 \psi(t,u)) \Psi(t,u),\\
            \partial_t V(t,u) &=& - \frac{1}{4} \sigma^2 \Psi(t,u)^2,
        \end{array}
    \right.
\end{equation}
with terminal values $\psi(T,u) = \phi(T,u) = V(T,u) = 0$ and $\Psi(T,u) = 1$.

Through a straightforward calculations, we can derive the conditional Laplace transform for the Heston model. The detailed calculation
is given in Appendix \ref{sec:app2}.
\begin{cor}
For the Heston model, the conditional Laplace transform of $Y_T$ given
$X_T = x_{\scriptscriptstyle T} > 0$ satisfies
\begin{eqnarray}
    \lefteqn{\pE_{x,y}\Big[e^{- u Y_T}\big| X_T = x_{\scriptscriptstyle T}\Big]}\nonumber\\
    &=& \frac{\eta(u)(1 - e^{-\kappa T})}{\kappa(1 - e^{-\eta(u) T})}
        \exp\Big\{ \textstyle - u \big(y  + (r - \frac{\kappa\theta\rho}{\sigma})T\big) - \frac{1}{2}(\eta(u) - \kappa)T\Big\}\nonumber\\
    & & \times \exp\Big\{\textstyle - \frac{1}{\sigma^2}
        \Big( \frac{\eta(u)(1 + e^{- \eta(u) T}) - (\kappa + u \sigma \rho)(1 - e^{- \eta(u) T})}{1 - e^{-\eta(u) T}}
        - \frac{2\kappa e^{ - \kappa T}}{1 - e^{- \kappa T}} \Big) x \Big\}\label{eq:condChfH}\\
    & & \times \exp\Big\{\textstyle - \frac{1}{\sigma^2}
        \Big( \frac{\eta(u)(1 + e^{- \eta(u) T}) + (\kappa + u \sigma \rho)(1 - e^{- \eta(u) T})}{1 - e^{-\eta(u) T}}
        - \frac{2\kappa}{1 - e^{- \kappa T}} \Big) x_{\scriptscriptstyle T} \Big\}\nonumber\\
    & & \times \frac{I_{0.5\delta - 1}\big[ \sqrt{x x_{\scriptscriptstyle T}}
        \frac{4\eta(u)e^{-0.5\eta(u)T}}{ \sigma^2(1-e^{-\eta(u)T}) }\big] }
        {I_{0.5\delta - 1}\big[ \sqrt{x x_{\scriptscriptstyle T}} \frac{4\kappa e^{-0.5\kappa T}}{ \sigma^2(1-e^{-\kappa T}) } \big] },\nonumber
\end{eqnarray}
where $I_{\nu}(\cdot)$ is the modified Bessel function of the first kind, and $\eta(u) = \sqrt{(\kappa + u \sigma \rho)^2 - \sigma^2 u(u+1)}$.
\end{cor}

\section{Review on Broadie and Kaya Method}\label{sec:BK}
Before going into our exact simulation method, we briefly review the Broadie and Kaya's exact simulation method \cite{BK06} of the Heston model.
The simulation method, they devised, motivated our research and it contains important idea of Fourier inversion techniques.

Since the system (\ref{eq:hest}) of the stochastic differential equations does not admit a closed-form solution and there is no
elementary way to simulate $X_T$ and $Y_T$ exactly. Broadie and Kaya proposed an exact sampling scheme which uses the Fourier
inversion techniques \cite{BK06}. From (\ref{eq:hest}),
\begin{eqnarray*}
    Y_T &=& y + r T - \frac{1}{2}\int_0^T X_s ds + \sqrt{1-\rho^2} \int_0^T \sqrt{X_s}dZ_s + \rho \int_0^T \sqrt{X_s} dW_s\\
    &=& y + r T - \frac{1}{2}\int_0^T X_s ds + \sqrt{1-\rho^2} \int_0^T \sqrt{X_s}dZ_s
        + \frac{\rho}{\sigma}\Big(X_T - x - \kappa \theta T + \kappa \int_0^T X_s ds \Big)\\
    &=& y + \frac{\rho}{\sigma}\big(X_T - x\big) + (r - \frac{\kappa\theta\rho}{\sigma})T
        + \big(\frac{\rho\kappa}{\sigma} - \frac{1}{2} \big)\int_0^T X_s ds + \sqrt{1-\rho^2}\int_0^T \sqrt{X_s} dZ_s.
\end{eqnarray*}
Recall that $Z$ is independent of $W$. Consequently, $Z$ is independent of the sigma field $\mathcal{F}^X_T = \sigma(X_t : 0 \le t \le T)$. So,
$\mathcal{F}^X_T$-conditional distribution of $\int_0^T \sqrt{X_s} dZ_s$ is normal with mean $0$. And its conditional variance can be computed
by the It\^{o}'s isometry:
\begin{equation*}
    \var_{x,y}\Big({\textstyle \int_0^T \sqrt{X_s} dZ_s}\Big|\mathcal{F}^X_T\Big)
        = \pE_{x,y}\Big[\Big({\textstyle \int_0^T \sqrt{X_s} dZ_s}\Big)^2\Big|\mathcal{F}^X_T\Big]
        = \pE_{x,y}\Big[{\textstyle \int_0^T X_s ds}\Big|\mathcal{F}^X_T\Big] = \int_0^T X_s ds. \\
\end{equation*}
These observations gives the following exact sampling scheme of the state variables $(X_T, Y_T)$:

\begin{alg}[Broadie and Kaya \cite{BK06}]\label{algo:BK-exact}
    This algorithm generates the state variables $X_T$ and $Y_T$ of the Heston model.
    \begin{enumerate}[Step (1)]
        {\setlength\itemindent{25pt}\item Generate a sample from the distribution of $X_T$ \label{step:BK-X}}
        {\setlength\itemindent{25pt}\item Generate a sample from the conditional distribution of $I = \int_0^T X_t dt$ given $X_T$ \label{step:BK-cond}}
        {\setlength\itemindent{25pt}\item Generate a standard normal random number $Z$ \label{step:BK-normal}}
        {\setlength\itemindent{25pt}\item Set $Y_T = y + \frac{\rho}{\sigma}\big(X_T - x\big) + (r - \frac{\kappa\theta\rho}{\sigma})T
            + \big(\frac{\rho\kappa}{\sigma} - \frac{1}{2} \big)I + \sqrt{(1-\rho^2)I} Z.$}
    \end{enumerate}
\end{alg}
For the step (\ref{step:BK-X}), one need to generate a sample from the distribution of $X_T$. Fortunately, the distribution of $X_T$ is well-known
in the literature (e.g. see Glasserman \cite{Gl04}). The law of $X_T$ can be expressed as
\begin{equation*}
    X_T = \frac{\sigma^2(1 - e^{-\kappa T})}{4\kappa} \chi'^2_{\delta}\Big( \frac{4\kappa e^{-\kappa T}}{\sigma^2(1-e^{-\kappa T})}x\Big),
\end{equation*}
where $\delta = \frac{4\kappa \theta}{\sigma^2}$, $\chi'^2_{\delta}(\lambda)$ denotes the noncentral chi-squared random variable with $\delta$ degrees of freedom, and a noncentrality parameter $\lambda$.
For the sampling of noncentral chi-squared random variable, refer to \cite{Gl04}.

The Broadie and Kaya's idea of Fourier inversion techniques comes into play at the step (\ref{step:BK-cond}). They showed that the conditional
characteristic function $\varphi(\cdot|x,x_{\scriptscriptstyle T})$ of $\int_0^T X_t dt$ given $X_T = x_{\scriptscriptstyle T}$ can be written as
\begin{eqnarray*}
    \lefteqn{\varphi(\lambda|x,x_{\scriptscriptstyle T})
        = \pE_{x,y}\Big[ \exp\big({\textstyle i \lambda \int_0^T X_t dt}\big)\Big| X_T = x_{\scriptscriptstyle T}\Big]}\\
    &=& \frac{\gamma(\lambda) e^{ -(1/2)(\gamma(\lambda) - \kappa)T}(1 - e^{-\kappa T})}{\kappa(1-e^{-\gamma(\lambda) T})}\\
    & & \times \exp\Big\{\frac{x+x_{\scriptscriptstyle T}}{\sigma^2} \Big[\frac{\kappa(1+e^{-\kappa T})}{1 - e^{-\kappa T}}
        - \frac{\gamma(\lambda)(1 + e^{-\gamma(\lambda) T})}{1 - e^{-\gamma(\lambda)T}}\Big]\Big\}
        \frac{I_{0.5\delta - 1}\big[ \sqrt{x x_{\scriptscriptstyle T}}
        \frac{4\gamma(\lambda)e^{-0.5\gamma(\lambda)T}}{ \sigma^2(1-e^{-\gamma(\lambda)T}) }\big] }
        {I_{0.5\delta - 1}\big[ \sqrt{x x_{\scriptscriptstyle T}} \frac{4\kappa e^{-0.5\kappa T}}{ \sigma^2(1-e^{-\kappa T}) } \big] },
\end{eqnarray*}
where $\gamma(\lambda) = \sqrt{\kappa^2 - 2\sigma^2 i \lambda}$. Then
they numerically inverted the conditional characteristic function $\varphi(\cdot |x,x_{\scriptscriptstyle T})$ to have an approximation of the
conditional distribution,
\begin{equation}\label{eq:BKdist}
    \pP_x\Big({\textstyle \int_0^T X_t dt \le v}|X_T = x_{\scriptscriptstyle T}\Big) \approx F_{h,N}(v |x,x_{\scriptscriptstyle T})
        = \frac{h v}{\pi} + \frac{2}{\pi} \sum_{n = 1}^N \frac{\sin(h n v)}{n} \re \big[ \varphi(h n|x,x_{\scriptscriptstyle T})\big],
\end{equation}
where $h > 0$ is the discretization step size.
Finally, they applied the inverse transform method to $F_{h,N}( \cdot |x,x_{\scriptscriptstyle T})$ to simulate the integral $\int_0^T X_t dt$,
i.e., they generated a uniform random number $U$ and solved the following equation for $\int_0^T X_t dt$ numerically
\begin{equation*}
    F_{h,N}\Big( {\textstyle \int_0^T X_t dt} |x,x_{\scriptscriptstyle T}\Big) = U.
\end{equation*}

\section{Exact Simulation Method}\label{sec:exact}
In this section, we present in detail our exact sampling algorithm of the WMSV model. Since the process $(X, Y)$ is a time-homogeneous Markov process,
the exact simulation method of $(X_T, Y_T)$ given an initial value $(X_0, Y_0) = (x,y)$ can be extended to the exact simulation method
of $(X_{t_2}, Y_{t_2})$ given $(X_{t_1}, Y_{t_2})$ for arbitrary $t_1 < t_2$ in a straightforward manner. Therefore we only consider the simulation
of $(X_T, Y_T)$ given an initial value $(X_0, Y_0) = (x,y)$.

For the WMSV model, a naive extension of Broadie and Kaya method is hardly accomplishable. The difficulty comes from
the dimensionality of the volatility factor process $X$. In the Broadie and Kaya method, one need to generate a sample from the conditional (univariate)
distribution of $\int_0^T X_t dt$ given $X_T = x_{\scriptscriptstyle T}$, and it can be achieved by Fourier inversion techniques. In contrast,
the integrated volatility factor $\int_0^T X_t dt$ of the WMSV model has a $d \times d$ matrix-variate distribution, which makes it almost
impossible to generate a sample from the distribution by Fourier inversion techniques. Hence, instead of following
the Broadie and Kaya's approach, we take a rather direct way to
achieve the goal through Theorem \ref{thm:main}. Roughly, our method consists of the following two step.

\begin{alg}\label{algo:exact}
    This algorithm generates $L$ pairs of the state variables $(X_T, Y_T)$ of the WMSV model.
    \begin{enumerate}[Step (1)]
        {\setlength\itemindent{25pt}\item Generate $L$ samples $X^{(1)}_T, \cdots, X^{(L)}_T$ from the distribution of $X_T$}
        {\setlength\itemindent{25pt}\item For each $l = 1, \cdots, L$, generate a sample $Y^{(l)}_T$
            from the conditional distribution of $Y_T$ given $X_T = X^{(l)}_T$}
    \end{enumerate}
\end{alg}
\noindent In the following subsections, we go through the details of these two steps.

\subsection{Sampling from the Distribution of $X_T$}
As indicated in the proof of Theorem \ref{thm:main}, $X_T$ has noncentral Wishart distribution with degrees of freedom $\delta$,
covariance matrix $V(0,0)$, and matrix of noncentrality parameter $V(0,0)^{-1}$ $\Psi(0,0)^{\top}$ $x$ $\Psi(0,0)$.
The Wishart distributions with integer degrees of freedom
(i.e., $\delta \in \mathbb{N}$) are extensively studied in the literature of the multivariate statistical analysis (e.g. see \cite{Gl76,Ks59,Mu82}).

In case that $\delta$ is an integer which is greater than or equal to $d$, one way of exact simulation is
squaring a normal random matrix \cite{Br91}. Let $N_t$ be a
solution of the following equation
\begin{equation}\label{eq:linearSDE}
    dN_t = N_t H^{\top} dt + dB_t \Sigma , \hspace{0.5cm} \text{ with } \hspace{0.5cm} N_0^{\top}N_0 = x,
\end{equation}
where $B$ is a standard $\delta \times d$ matrix Brownian motion. One can easily check that $N^{\top} N$ satisfies the stochastic differential
equation (\ref{eq:wishart}). The equation (\ref{eq:linearSDE}) admits a closed form solution
\begin{equation*}
    N_t = \left( N_0 + \int_0^t dB_s \Sigma e^{- s H} \right) e^{t H}, \hspace{0.5cm} \text{ for } \hspace{0.5cm} 0 \le t < \infty.
\end{equation*}
Notice that the rows of $N_T$ are independent normal random vectors with common covariance matrix $V(0,0)$ and the mean of $N_T$ is $N_0 e^{T H}$.
Therefore, for integer degrees $\delta$ of freedom, the exact sampling of $X_T$ can be achieved by sampling $\delta \times d$ i.i.d. normal random variables. \
And there is also a more sophisticated but efficient way of simulating noncentral Wishart distributions with an integer degrees of freedom \cite{Gl76}.

As far as we know, an exact sampling scheme for a noncentral Wishart distribution with non-integer valued degrees of freedom has only recently been
devised by Ahdida and Alfonsi \cite{AA10}. They have used a splitting method of the infinitesimal generator of Wishart processes. Their exact
sampling scheme requires sampling of at most $d(d-1)/2$ i.i.d. normal random variables and $d$ noncentral chi-square random variables.
In the numerical experiments of this paper, we used the exact sampling method of Ahdida and Alfonsi \cite{AA10}.

\subsection{Sampling from the Conditional Distribution of $Y_T$ Given $X_T$}\label{sec:condSamp}
This subsection is devoted to the step of sampling from the conditional distribution of $Y_T$ given $X_T = x_{\scriptscriptstyle T}$. This
step is the most technical and time consuming step in our method. As explained in Section \ref{sec:BK}, Broadie and Kaya adopted Fourier inversion
techniques to invert the conditional characteristic function of $\int_0^T X_t dt$ given $X_T = x_{\scriptscriptstyle T}$. We follow a similar
approach, but we invert directly the conditional characteristic function of $Y_T$ given $X_T = x_{\scriptscriptstyle T}$ to avoid the difficulty
in converting the characteristic function of matrix-variate random variable.

Let $\varphi(\cdot; x, y, x_{\scriptscriptstyle T})$ be the conditional characteristic function, i.e.,
\begin{equation*}
    \varphi(\lambda; x, y, x_{\scriptscriptstyle T}) = \pE_{x,y}\Big[ e^{i\lambda Y_T} \Big| X_T = x_{\scriptscriptstyle T}\Big],
    \hspace{0.5cm} \text{ for } \hspace{0.5cm} \lambda \in \reals.
\end{equation*}
And let $F(\cdot; x, y, x_{\scriptscriptstyle T})$ be the corresponding distribution function:
\begin{equation*}
    F(v; x, y, x_{\scriptscriptstyle T}) = \pP_{x,y}\Big( Y_T \le v \big| X_T = x_{\scriptscriptstyle T}\Big).
\end{equation*}
By Levy's inversion formula, the distribution can be recovered from the characteristic function, i.e., for $-\infty < l_{\epsilon} < v$, we have
\begin{eqnarray*}
    F(v; x, y, x_{\scriptscriptstyle T}) &=& F(l_{\epsilon}; x, y, x_{\scriptscriptstyle T})
        + \frac{1}{2\pi}\int_{-\infty}^{\infty} \varphi(\lambda; x, y, x_{\scriptscriptstyle T})
        \frac{e^{-i\lambda l_{\epsilon}} - e^{-i\lambda v}}{i\lambda} d\lambda\\
    &=& F(l_{\epsilon}; x, y, x_{\scriptscriptstyle T}) + \frac{1}{\pi} \int_0^{\infty}
        \re\Big[ \varphi(\lambda; x, y, x_{\scriptscriptstyle T})\frac{e^{-i\lambda l_{\epsilon}} - e^{-i\lambda v}}{i\lambda} \Big] d\lambda\\
    &=& F(l_{\epsilon}; x, y, x_{\scriptscriptstyle T}) + \frac{1}{\pi} \int_0^{\infty}
        \im\Big[ \varphi(\lambda; x, y, x_{\scriptscriptstyle T})(e^{-i\lambda l_{\epsilon}} - e^{-i\lambda v}) \Big] \frac{d\lambda}{\lambda}.
\end{eqnarray*}

Notice that we can make $F(l_{\epsilon};x,y,x_{\scriptscriptstyle T})$ small enough to ignore by taking a small $l_{\epsilon}$. So the integral
terms are dominant in the above expressions. The integral above can be approximated by a numerical integration method. We use the trapezoidal rule to compute the distribution function
numerically. It is known that the trapezoidal rule works well for oscillating integrands because the errors tend to be cancelled (see section 4 of
Abate and Whitt \cite{AW92}). An integral on the whole positive real line can be approximated in the following way. For notational simplicity, we write the
integrand as $g(\lambda)$.
\begin{equation*}
    \int_0^{\infty} g(\lambda) d \lambda = \frac{g(0+)}{2} h   + h \sum_{n = 1}^{\infty} g(n h) - e_d(h)
        = \frac{ g(0+)}{2} h + h \sum_{n = 1}^{N} g(n h) - e_d(h) - e_t(N),
\end{equation*}
where $e_d(h)$ and $e_t(N)$ are errors due to the discretization of the continuous variable and the truncation of the infinite sum, respectively. Each of these errors can be made arbitrarily small by
taking sufficiently small $h$ and sufficiently large $N$, and we give a detailed discussion on the control of these errors in Section \ref{sec:issue}.
One can easily find that $g(0+) = v - l_{\epsilon}$. Therefore, we approximate the distribution
function by the following finite trigonometric series:
\begin{eqnarray}
    F_{h, N}(v; x, y, x_{\scriptscriptstyle T}) &=& \frac{h(v - l_{\epsilon})}{2\pi}
    + \frac{1}{\pi} \sum_{n=1}^{N}
        \im\Big[ \frac{\varphi(n h; x, y, x_{\scriptscriptstyle T})}{n}(e^{-i n h l_{\epsilon}} - e^{-i n h v}) \Big]\nonumber\\
    &=& \frac{h(v - l_{\epsilon})}{2\pi} + \frac{1}{\pi} \sum_{n=1}^{N}
        \Big( \frac{\re\big[\varphi(n h; x, y, x_{\scriptscriptstyle T})]}{n } \big(\sin(n h v) - \sin(n h l_{\epsilon})\big)\Big)\label{eq:dist}\\
    & & - \frac{1}{\pi} \sum_{n=1}^{N}
        \Big(\frac{\im\big[\varphi(n h; x, y, x_{\scriptscriptstyle T})]}{n }\big(\cos(n h v) - \cos(n h l_{\epsilon})\big)\Big).\nonumber
\end{eqnarray}

We address how to evaluate the conditional characteristic function $\varphi(\lambda;x,y,x_{\scriptscriptstyle T})$. The conditional
characteristic function can be obtained by taking $u = - i \lambda$ in the formula (\ref{eq:condChf}).
The formula (\ref{eq:condChf}) involves the solutions $\psi$, $\phi$, $\Psi$, and $V$ of the system (\ref{eq:equations}) of ordinary differential equations.
As we have shown in Section \ref{sec:LapHes}, the system (\ref{eq:equations}) admits a closed-form solution for the Heston model.
In general, the system (\ref{eq:equations}) does not admit a closed-form solution because of non-commutativity of matrix multiplications.
But the system (\ref{eq:equations}) can be efficiently solved by numerical methods. In particular, we used the MATLAB function
{\sf ode45} for solving equations in our experiments.\footnote{The MATLAB function {\sf ode45} is based on an explicit Runge-Kutta $(4,5)$ formula of Dormand and
 Prince\cite{DP80}.} It should be noted that the functions $\psi$, $\phi$, $\Psi$, and $V$ are needed to be evaluated only at
$(0, 0), (0, -i h), \cdots, (0, - i N h)$, and such evaluation points can be chosen uniformly across samples $X^{(1)}_T, \cdots, X^{(L)}_T$.
Therefore, it is enough to solve the system (\ref{eq:equations}) at those grid points only once
for all simulation runs, and this numerical step does not cause a computational burden if the number $L$ of simulation runs is relatively larger than $N$.

The expression (\ref{eq:condChf}) requires calculations of the power of complex numbers, which is, in general, a multi-valued function:
\begin{equation*}
    (z)^{\nu} = (|z|e^{i\Argu(z)})^{\nu} = (|z|e^{i(\Argu(z)+2 m \pi)})^{\nu} = |z|^{\nu} e^{i\nu\Argu(z)+2 m\nu \pi}, \hspace{0.3cm} m \in \mathbb{Z},
\end{equation*}
with principal argument $- \pi < \Argu(z) \le \pi$. If $\lambda \mapsto z(\lambda)$ is a complex-valued continuous function on $\reals$ which
does not attains $0$, then there exists a unique continuous function $\lambda \mapsto \argu(z(\lambda))$ such that $- \pi < \argu(z(0)) \le \pi$.
Such a continuation can be easily constructed by tracing the principal argument $\Argu(z(\lambda))$ and adding or subtracting $2\pi$ at the discontinuous
points. And then, we can use it to construct a continuous version of the power function:
\begin{equation*}
    (z(\lambda))^{\nu} = |z(\lambda)|^{\nu} e^{i\nu\argu(z(\lambda))}.
\end{equation*}
It is obvious that $\lambda \mapsto \dete V(0, -i\lambda)$ is a continuous function which does not attains $0$, and we can apply the above observation
to calculate $(\dete V(0, -i\lambda))^{\delta/2}$.

The conditional characteristic function (\ref{eq:condChf}) involves a hypergeometric function of matrix argument. The hypergeometric function
of matrix argument is defined by the following power series
\begin{equation}\label{eq:hyper}
    {}_0 F_1(b; y) = \sum_{k=0}^{\infty} \sum_{|\iota| = k} \frac{1}{(b)_{\iota}} \frac{C_{\iota}(y)}{k!}.
\end{equation}
For the definitions of $|\iota|$ and $(b)_{\iota}$, see Section \ref{sec:def}. The zonal polynomials $C_{\iota}(y)$ are not
polynomials of the matrix $y$, but polynomials of its eigenvalues. So, we sometimes write the hypergeometric functions and
the zonal polynomials in the following way:
\begin{equation*}
    C_{\iota}(y) = C_{\iota}(\alpha_1, \cdots, \alpha_d),\hspace{0.3cm}
    {}_0 F_1(b;y) = {}_0 F_1 (b;\alpha_1, \cdots, \alpha_d), \hspace{0.3cm}
\end{equation*}
where $\alpha_1, \cdots, \alpha_d$ are eigenvalues of $y$. Koev and Edelman \cite{KE06} exploited the combinatorial properties of zonal
polynomials and the generalized hypergeometric coefficients to develop an algorithm for computing a truncated version of (\ref{eq:hyper}):
\begin{equation*}
    {}_0^m F_1^{ }(b; \alpha_1, \cdots, \alpha_d)
        = \sum_{k=0}^{m} \sum_{|\iota| = k} \frac{1}{(b)_{\iota}} \frac{C_{\iota}(\alpha_1, \cdots, \alpha_d)}{k!}.
\end{equation*}
Since the denominator of the series (\ref{eq:hyper}) grows faster than factorial order, the series converges quickly and the truncated version gives
a good approximation. We provide a detailed error analysis of the truncated series in Section \ref{sec:issue}.
MATLAB implementations of Koev and Edelman's algorithms are available in the author's homepage \cite{KoWeb}. But the routine gets only real eigenvalues
as input arguments. So we modified the codes applicable for complex input arguments.

To sample the log-price $Y_T$, we use the inverse transform method. We generate a uniform random number $U$
and we numerically solve the equation for $Y_T$
\begin{equation}\label{eq:Newton}
    F_{h,N}(Y_T; x, y, x_{\scriptscriptstyle T}) = U.
\end{equation}
The equation can be solved efficiently by numerical method, e.g., Newton's method, because the function $F_{h,N}$ is a strictly increasing function.
We provide details in the following subsection.

\subsection{Some Implementation Issues}\label{sec:issue}
In order to implement our method, we need to resolve some issues. We have to determine appropriate values of $l_{\epsilon}$, $h$, and $N$ in (\ref{eq:dist}),
and decide the number of summands to approximate the infinite series (\ref{eq:hyper}). We also need to address how to solve the equation (\ref{eq:Newton}).

A careful choice of values of $\l_{\epsilon}$, $h$, and $N$ is crucial in our method, because they determine
the points where the characteristic function $\varphi(\cdot;x,y,x_{\scriptscriptstyle T})$ is evaluated.
Too few evaluation points might introduce a large bias of the Monte Carlo simulation due to the approximation error of the conditional distribution function.
Too many points might make our method too slow, because the evaluation of the characteristic function
$\varphi(\cdot;x,y,x_{\scriptscriptstyle T})$ is the most time-consuming step of our method.

As indicated in the previous subsection, the approximation (\ref{eq:dist}) involves three different kinds of errors: the discretization error $e_d(h)$,
the truncation error $e_t(N)$, and the left tail probability $F(l_{\epsilon};x,y,x_{\scriptscriptstyle T})$. We suggest a way to control these errors based on
the conditional mean and standard deviation of $Y_T$ given $X_T = x_{\scriptscriptstyle T}$. Recall that the mean and standard deviation can be
easily found by differentiating the characteristic function $\varphi(\cdot;x,y,x_{\scriptscriptstyle T})$:
\begin{eqnarray*}
    \mu(x_{\scriptscriptstyle T}) &=& \pE_{x,y}[Y_T | X_T = x_{\scriptscriptstyle T}] = \im\big(\varphi'(0;x,y,x_{\scriptscriptstyle T})\big),\\
    \sigma(x_{\scriptscriptstyle T})^2 &=& \pE_{x,y}[(Y_T - \mu)^2 | X_T = x_{\scriptscriptstyle T}]
        = -\re\big(\varphi''(0;x,y,x_{\scriptscriptstyle T})\big) - \big(\im(\varphi'(0;x,y,x_{\scriptscriptstyle T}))\big)^2.
\end{eqnarray*}
In the section 5 of Abate and Whitt \cite{AW92}, they used Poisson summation formula to prove that the discretization error is bounded by two
simple tail probabilities
\begin{equation*}
    0 \le e_d(h) \le \big(1 - F(2 \pi /h + l_{\epsilon};x,y,x_{\scriptscriptstyle T})\big) + F(-2 \pi/ h + v;x,y,x_{\scriptscriptstyle T})
\end{equation*}
for $v - l_{\epsilon} < 2\pi /h$. Suppose that we want to calculate $F(v;x,y,x_{\scriptscriptstyle T})$ with an error less than $\epsilon$. Let $l_{\epsilon}$ and
$u_{\epsilon}$ such that $F(l_{\epsilon};x,y,x_{\scriptscriptstyle T}) \le \epsilon/4$ and
$F(u_{\epsilon};x,y,x_{\scriptscriptstyle T}) \ge 1 - \epsilon/4$. Take $h = \frac{2\pi}{u_{\epsilon} - l_{\epsilon}}$. Then,
for $l_{\epsilon} \le v \le u_{\epsilon}$,
\begin{eqnarray*}
    \frac{2\pi}{h} + l_{\epsilon} = u_\epsilon, \hspace{0.3cm} \text{ and } \hspace{0.3cm}
        -\frac{2\pi}{h} + v \le -\frac{2\pi}{h} + u_{\epsilon} =  l_{\epsilon},
\end{eqnarray*}
so that
\begin{eqnarray*}
    |F(l_{\epsilon}; x,y,x_{\scriptscriptstyle T}) - e_d(h)| \le \max\{F(l_{\epsilon}; x,y,x_{\scriptscriptstyle T}), e_d(h)\} \le \epsilon/2.
\end{eqnarray*}
Then we turn to the truncation error $e_t(N)$. Since the summands in (\ref{eq:dist}) are oscillating, the absolute value of the last term gives an
estimate for the truncation error:
\begin{equation*}
    |e_t(N)| \approx \frac{1}{\pi} \Big|\im\Big[ \frac{\varphi(N h; x, y, x_{\scriptscriptstyle T})}{N}(e^{-i N h l_{\epsilon}} - e^{-i N h v})\Big]\Big|
    \le \frac{2|\varphi(N h; x, y, x_{\scriptscriptstyle T})|}{\pi}
\end{equation*}
So, we may choose a large $N$ so that $|\varphi(N h; x, y, x_{\scriptscriptstyle T})| \le \epsilon\pi/4$ to make the truncation error approximately less than
$\epsilon/2$.

These error bounds and estimates
are theoretically appealing, but they are not suitable for practical use. For example, we can compute
$\varphi(N h; x, y, x_{\scriptscriptstyle T})$, but it is very time-consuming. Furthermore, we need to evaluate
$F(\cdot;x,y,x_{\scriptscriptstyle T})$, which is not known in advance. To overcome these difficulties, we exploit a heuristic idea, which is
obtained from numerical experiments. The conditional distribution of $Y_T$ given $X_T = x_{\scriptscriptstyle T}$ is roughly similar to the normal
distribution, because the terminal value $X_T$ of volatility factor is already known and the main source of randomness is the stochastic integral
with respect to a Brownian motion. But they show some differences as well. The normal distribution is symmetric with respect to mean, but the conditional
distribution of $Y_T$ is asymmetric, and the decay rate of the conditional characteristic function of $Y_T$ gets slower as the correlation parameter
$R$ increases. So, we start with $l_{\epsilon}$, $h$, and $N$ which are suggested
by the normal distribution, and then modify them to take such effects into account. The normal distribution suggests us to take $l_{\epsilon}$, $h$,
and $N$ of the following form:
\begin{equation*}
    \left\{
        \begin{array}{ll}
            l_{\epsilon} &= \Phi^{-1}\Big(\epsilon/4; \mu(x_{\scriptscriptstyle T}), \sigma(x_{\scriptscriptstyle T})^2 \Big)
                = \mu(x_{\scriptscriptstyle T}) + \sigma(x_{\scriptscriptstyle T}) \Phi^{-1}(\epsilon/4),\\
            h &= 2 \pi \Big(\Phi^{-1}\big(1-\epsilon/4; \mu(x_{\scriptscriptstyle T}), \sigma(x_{\scriptscriptstyle T})^2\big)
                - \Phi^{-1}\big(\epsilon/4; \mu(x_{\scriptscriptstyle T}), \sigma(x_{\scriptscriptstyle T})^2\big)\Big)^{-1}\\
            &= - \pi \Big(\sigma(x_{\scriptscriptstyle T}) \Phi^{-1}\big(\epsilon/4\big)\Big)^{-1},\\
            N &= \left\lceil \frac{1}{h\sigma(x_{\scriptscriptstyle T})} \sqrt{- 2 \log(\pi\epsilon/4)} \right\rceil,
        \end{array}
    \right.
\end{equation*}
where $\Phi^{-1}(\cdot; \mu, \sigma^2)$ and $\Phi^{-1}(\cdot)$ are the inverse functions of the normal distribution $\mathcal{N}(\mu, \sigma^2)$
and the standard normal distribution, respectively. And $\lceil a \rceil$ denotes the smallest integer greater than $a$. Then we modify them with the
mean $\mu(x_{\scriptscriptstyle T})$ and the correlation parameter $R$. Since we want to choose the points $u = - i n h$, $n = 1, \cdots, N$ uniformly
across samples $X^{(1)}_T, \cdots, X^{(L)}_T$, we take $l_{\epsilon}$, $h$, and $N$ in a conservative way:
\begin{eqnarray}\label{eq:choice}
    \left\{
        \begin{array}{ll}
            l_{\epsilon} &= \min\limits_{l = 1, \cdots, L}
                \left\{ \mu({\scriptstyle X^{(l)}_T}) + \sigma({\scriptstyle X^{(l)}_T}) \Phi^{-1}(\epsilon/4) \right\},\\
            h &= -\pi \times \min\limits_{l = 1, \cdots, L}
                \Big(\big(c_1 |\mu({\scriptstyle X^{(l)}_T})| +  \sigma({\scriptstyle X^{(l)}_T}) \big)\Phi^{-1}\big(\epsilon/4\big)\Big)^{-1} ,\\
            N &= \max\limits_{l = 1, \cdots, L} \left\lceil \frac{1 + c_2 \sqrt{\tr[R^{\top}R]}}
                {h\sigma({\scriptstyle X^{(l)}_T})} \sqrt{- 2 \log(\pi\epsilon/4)} \right\rceil,
        \end{array}
    \right.
\end{eqnarray}
for appropriate constants $c_1 \ge 0$ and $c_2 \ge 0$. In particular, we take $c_1 = 0.1$ and $c_2 = 0.5$ in our experiments.

Now, we consider the problem of determining where to truncate the hypergeometric series (\ref{eq:hyper}).
The partitions of integers can be ordered lexicographically, i.e., if $\iota = (k_1, \cdots, k_d)$ and $\tilde{\iota} = (l_1, \cdots, l_d)$
are two partitions of integers we will write $\iota > \tilde{\iota}$ if $k_i > l_i$ for
the first index $i$ at which the parts become unequal. With this order relation, we give an upper bound for the truncation error of the hypergeometric
functions. The proof is given in Appendix \ref{sec:app1}.

\begin{prop}\label{prop:hyperub}
For $b > \frac{1}{2}(d-1)$ and an integer $m \ge d-1$, we have for $\alpha_1, \cdots, \alpha_d \in \complex$
\begin{equation}\label{eq:hyperub}
    |{}_0 F_1(b;\alpha_1, \cdots, \alpha_d) - {}_0^m F_1^{}(b;\alpha_1, \cdots, \alpha_d)|
        \le \frac{(|\alpha_1|+\cdots+|\alpha_d|)^{m+1}}{(m+1)!(b)_{\hat{\iota}(m+1)}} e^{|\alpha_1|+\cdots+|\alpha_d|},
\end{equation}
where $\hat{\iota}(m+1)$ is the smallest partition among the partitions of $m+1$ into not more than $d$ parts.
\end{prop}

Notice that the smallest partition $\hat{\iota}(k) = (k_1, \cdots, k_d)$ among the partitions of $k$ into not more than $d$ parts is the most evenly
distributed partition, and the smallest partition $\hat{\iota}(k+1)$ of $k+1$ is the partition which is obtained by adding $1$ to $k_1$ if
all the $k_i$'s are equal, or adding $1$ to $k_j$ at which $k_j$ differs from $k_{j-1}$ otherwise. For example, $(2,2,2,2,1)$, $(2,2,2,2,2)$,
and $(3,2,2,2,2)$ are the smallest partitions of $9$, $10$, and $11$ into $5$ parts, respectively. With this observation, we can compute the hypergeometric
coefficients $(b)_{\hat{\iota}(k)}$ as follows: $(b)_{(1)} = b$ and
\begin{equation}\label{eq:hyperco}
    (b)_{\hat{\iota}(k+1)} =
    \left\{
        \begin{array}{ll}
            (b)_{\hat{\iota}(k)} (b + k_1) &\hspace{0.5cm} \text{ if } k_1 = \cdots = k_d\\
            (b)_{\hat{\iota}(k)} \big(b + k_j - \frac{1}{2}(j-1) \big)
                &\hspace{0.5cm} \text{ if } k_1 = \cdots = k_{j-1} \neq k_{j} = \cdots = k_d
        \end{array}
    \right.,
\end{equation}
where $\hat{\iota}(k) = (k_1, \cdots, k_d)$.

We use the bound (\ref{eq:hyperub}) as a criterion for truncating the series (\ref{eq:hyper}). Notice that
the parameter $b$ is a constant $\frac{1}{2} \delta$, which satisfies the assumption of Proposition \ref{prop:hyperub}. Since it is a constant,
we need to compute the hypergeometric coefficients of the smallest partitions only once in all the simulation runs, and it does not add computational
burden to our method. So the upper bound (\ref{eq:hyperub}) can be calculated with minor additional computational budget.

In order to solve the equation (\ref{eq:Newton}), we use the Newton's method. A careful choice of the initial guess is necessary to enhance the convergence.
As explained before,
the conditional distribution $F(v;x,y,x_{\scriptscriptstyle T})$ is approximated as the normal one with mean $\mu(x_{\scriptscriptstyle T})$
and variance $\sigma(x_{\scriptscriptstyle T})^2$. Therefore, we take the initial guess $v_0$ as the solution of the equation
\begin{equation*}
    \Phi(v_0; \mu(x_{\scriptscriptstyle T}), \sigma(x_{\scriptscriptstyle T})^2) = U, \hspace{0.5cm} \text{ or } \hspace{0.5cm}
    v_0 = \Phi^{-1}(U; \mu(x_{\scriptscriptstyle T}), \sigma(x_{\scriptscriptstyle T})^2).
\end{equation*}
Then we apply the Newton's method to the function $F_{h,N}(v;x,y,x_{\scriptscriptstyle T})$ to generate a sequence which approximates the
solution of the equation (\ref{eq:Newton}):
\begin{equation*}
    v_{k+1} = v_k - \frac{F_{h,N}(v_k;x,y,x_{\scriptscriptstyle T})}{F'_{h,N}(v_k;x,y,x_{\scriptscriptstyle T})}.
\end{equation*}
We iterate the loop for a fixed number of times, usually $4$ or $5$ times, and then use the bisection search method if the sequence, generated by Newton's method,
fails to converges within the predetermined number of iterations.

We summarize the method of generating a sample from the conditional distribution of $Y_T$ given $X_T = x_{\scriptscriptstyle T}$.
\begin{alg}\label{algo:integ}
    This algorithm generates $L$ samples $Y^{(1)}_T, \cdots, Y^{(L)}_T$ from the conditional distribution of $Y_T$ given
    $X_T =X^{(1)}_T, \cdots, X^{(L)}_T$, respectively.
    \begin{enumerate}[Step (1)]
        {\setlength\itemindent{25pt}\item Determine $l_{\epsilon}$, $h$, and $N$ according to (\ref{eq:choice}),}
        {\setlength\itemindent{25pt}\item Evaluate $\psi(0, u)$, $\phi(0,u)$, $\Psi(0,u)$,
            and $V(0,u)$ at $u = - i n h$, $n = 1, \cdots, N$}
        {\setlength\itemindent{25pt}\item Evaluate $\varphi(n h;x,y,X^{(l)}_T)$ for $n = 1, \cdots, N$ and $l = 1, \cdots, L$ }
        {\setlength\itemindent{25pt}\item Generate IID uniform random numbers $U^{(1)}, \cdots, U^{(L)}$}
        {\setlength\itemindent{25pt}\item For each $l = 1, \cdots, L$, solve $F_{h,N}(Y^{(l)}_T; x, y, X^{(l)}_T) = U^{(l)}$ for $Y^{(l)}_T$}
    \end{enumerate}
\end{alg}

\section{Numerical Results} \label{sec:numer}
In this section, we compare numerically our exact simulation method with other existing methods: the Euler discretization scheme
for the WMSV model and the Broadie-Kaya method for the Heston model.

\subsection{Comparison between the Exact Sampling and the Euler Scheme}

The system of stochastic differential equations (\ref{eq:wishart}) and (\ref{eq:logprice}) does not admit a closed form solution. In such a case,
one way to simulate the model is to discretize the time interval and simulate another process, which approximates the model, on these discrete time grids.
Euler discretization is a basic discretization scheme. Let $0 = t_0 < t_1 < \cdots < t_N = T$ be a partition of the time interval $[0, T]$ into
$N$ equal subintervals, i.e., $\Delta t = t_i - t_{i-1} = T/N$, $i = 1, \cdots, N$. We discretize the Wishart process (\ref{eq:wishart}) by setting
$\hat{X}_{t_0} = x$, and
\begin{equation*}
    \hat{X}_{t_i} = \left(\hat{X}_{t_{i-1}} + (\delta \Sigma^{\top}\Sigma + H \hat{X}_{t_{i-1}} + \hat{X}_{t_{i-1}} H^{\top}) \Delta t
        + {\textstyle \sqrt{\hat{X}_{t_{i-1}}}} \Delta W_{t_i} \Sigma
        + \Sigma^{\top} \big(\Delta W_{t_i}\big)^{\top}  {\textstyle \sqrt{\hat{X}_{t_{i-1}}}} \right)^+,
\end{equation*}
where $\Delta W_{t_i} = W_{t_i} - W_{t_{i-1}}$. Here, $A^+$ denotes the positive part of a symmetric matrix $A$: we set
$A^+ = O \diag(\lambda_1^+, \cdots, \lambda_d^+) O^{\top}$ if $A = O \diag(\lambda_1, \cdots, \lambda_d) O^{\top}$ and $O$ an orthogonal matrix.
In order to make $\hat{X}_{t_i}$ positive semidefinite, we take the positive part at each time grid. The discretization of the log-price process
(\ref{eq:logprice}) is
\begin{equation*}
    \hat{Y}_{t_i} = \hat{Y}_{t_{i-1}} + \Big(r - \frac{1}{2} \tr[\hat{X}_{t_{i-1}}]\Big) \Delta t
        + \tr\Big[ {\textstyle \sqrt{\hat{X}_{t_{i-1}}}} \Big(\Delta W_{t_i} R^{\top} + \Delta Z_{t_i} \sqrt{I_d - RR^{\top}}\Big)\Big],
\end{equation*}
where $\Delta Z_{t_i} = Z_{t_i} - Z_{t_{i-1}}$. Euler discretization scheme is easy to understand and implement. But it has serious disadvantages: the
distribution of the samples drawn from Euler scheme is different from the true distribution, and it may require very fine discretization to
reduce the bias as small as acceptable. These will be illustrated by numerical results.

We will compare the exact simulation method with the Euler discretization scheme in the terms of distributions, convergence and performance. The simulation
experiments in this paper were performed on a desktop PC with an Intel Core2 Quad $3.00$ GHz processor and $3.25$ GB RAM, running Windows XP Professional.
We used the MATLAB in the version R2009a.

In order to demonstrate the difference of empirical distributions of exact simulation method and Euler discretization method, we generated $10^6$ samples and
estimated the density functions of $Y_T$ using those samples. To estimate density functions, we used the MATLAB function {\sf ksdensity} which
computes a density estimate using kernel density estimation. The theoretical density function was obtained by a numerical
inversion of the characteristic function (\ref{eq:Laplace}). Figure \ref{figure:density} shows the density functions of $Y_T$ for the parameter set:
\begin{equation*}
    \delta = 1.1, \hspace{0.5cm} r = 0, \hspace{0.5cm} y = 0, \hspace{0.5cm} T = 1.0,
\end{equation*}
\begin{eqnarray*}
    x &=& \left[
            \begin{array}{cc}
                0.0298 & 0.0119\\
                0.0119 & 0.0108
            \end{array}
        \right], \hspace{0.5cm}
    H = \left[
            \begin{array}{cc}
                - 1.2479 & - 0.8985\\
                - 0.0820 & - 1.1433
            \end{array}
        \right],\\
    \Sigma &=& \left[
            \begin{array}{cc}
                0.3417 & 0.3493\\
                0.1848 & 0.3090
            \end{array}
        \right], \hspace{0.5cm}
    R = \left[
            \begin{array}{cc}
                -0.2243 & -0.1244\\
                -0.2545 & -0.7230
            \end{array}
        \right].\\
\end{eqnarray*}
The set of parameters except $\delta$ for the model is taken from Da Fonseca and Grasselli\footnote{In their paper, the calibrated $\delta$
for DAX index is $0.5776$, and it violate the assumption $\delta > d-1$. Indeed, in that case, the Wishart process $X$ is no longer an affine process
in the sense of Cuchiero et. al. \cite{CFMT11}.} \cite{DG11}. The number of time steps for the Euler method is set to $25$. From Figure
\ref{figure:density}, we observe that the distribution of the samples drawn from Euler method apparently differs from the true distribution.

\begin{figure}
        \centering
        \begin{subfigure}[b]{0.45\textwidth}
                \centering
                \includegraphics[width=\textwidth]{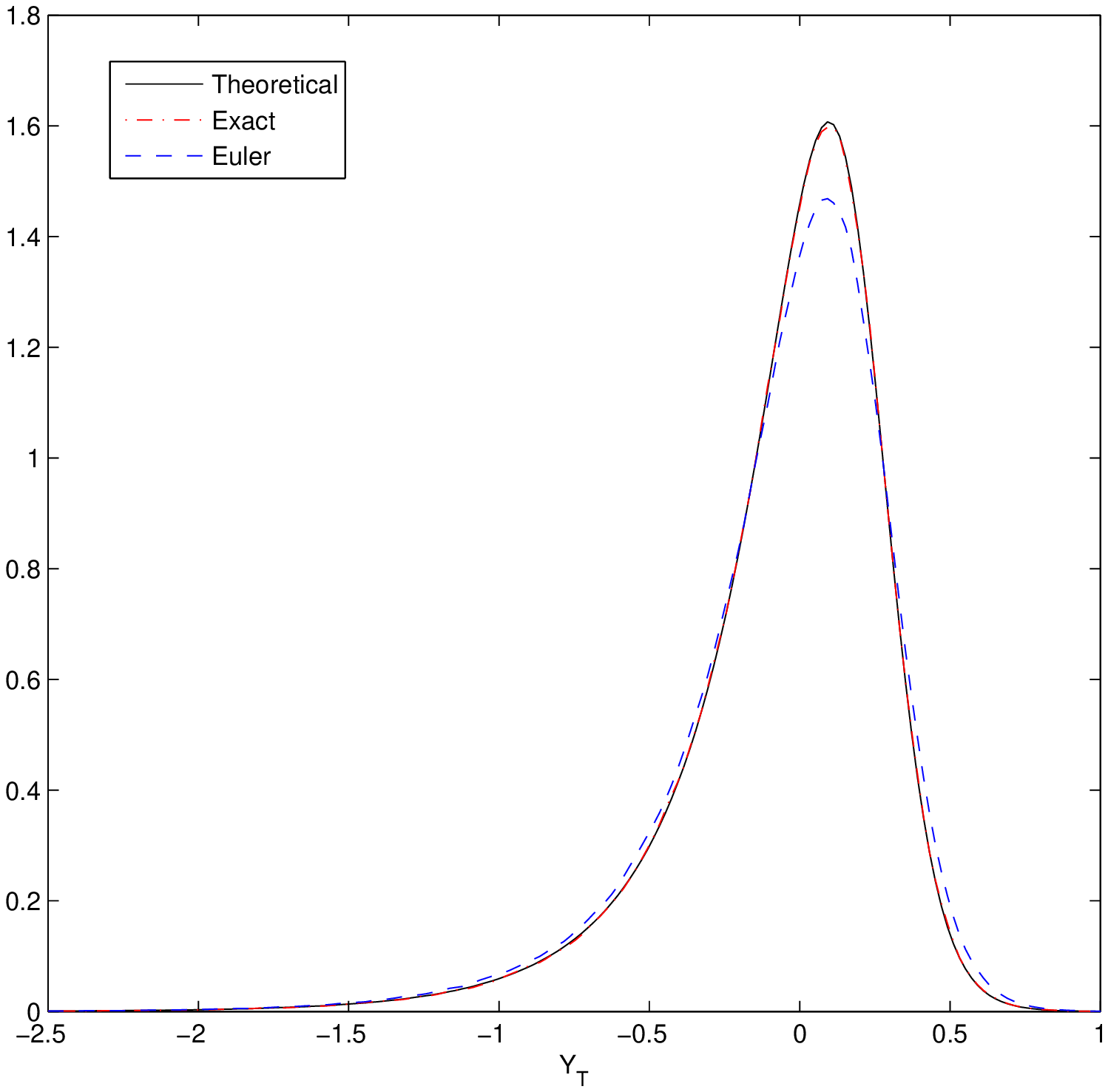}
        \end{subfigure}
        \quad
        \begin{subfigure}[b]{0.45\textwidth}
                \centering
                \includegraphics[width=\textwidth]{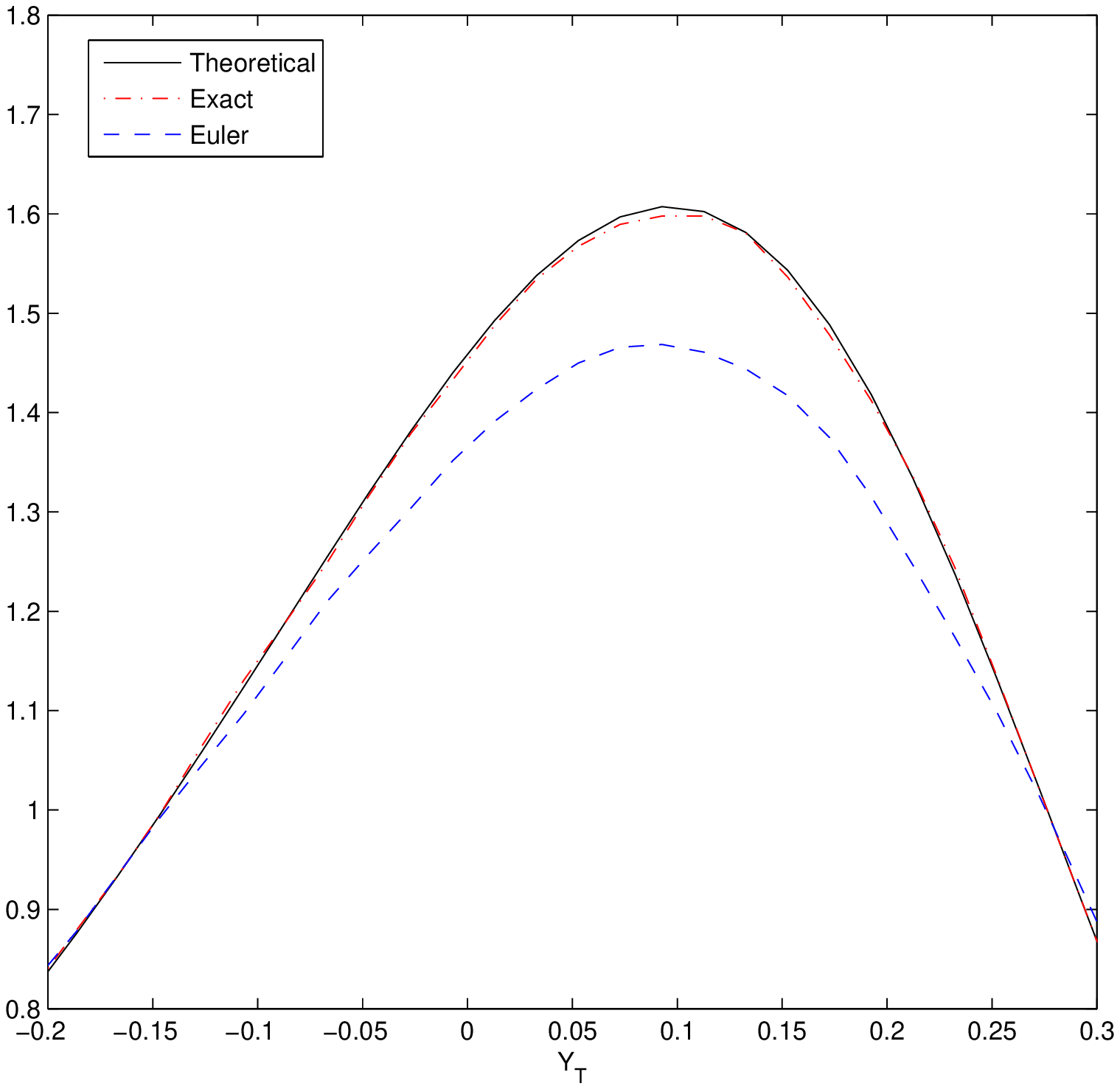}
        \end{subfigure}
        \\
        \begin{subfigure}[b]{0.45\textwidth}
                \centering
                \includegraphics[width=\textwidth]{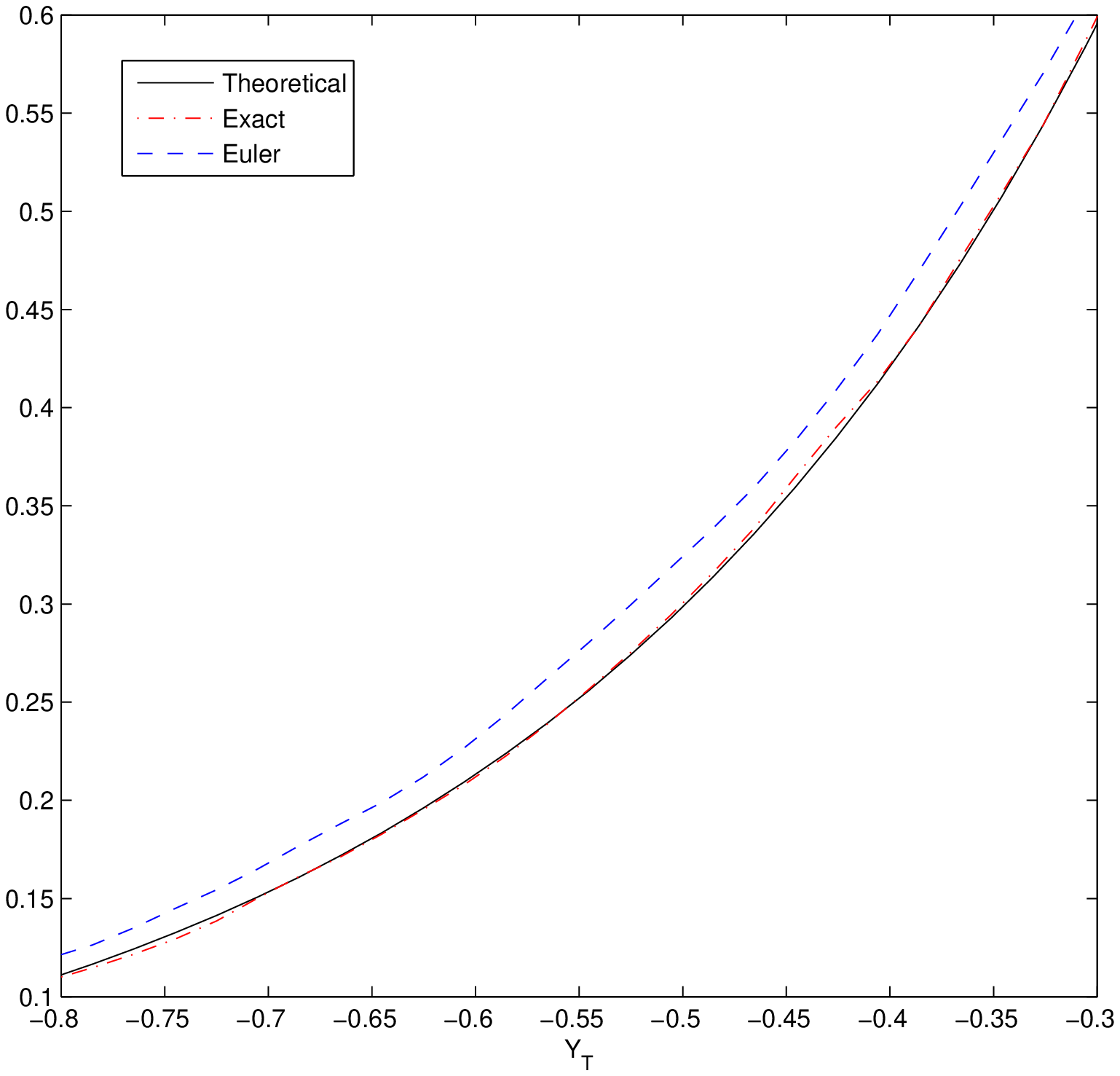}
        \end{subfigure}
        \begin{subfigure}[b]{0.45\textwidth}
                \centering
                \includegraphics[width=\textwidth]{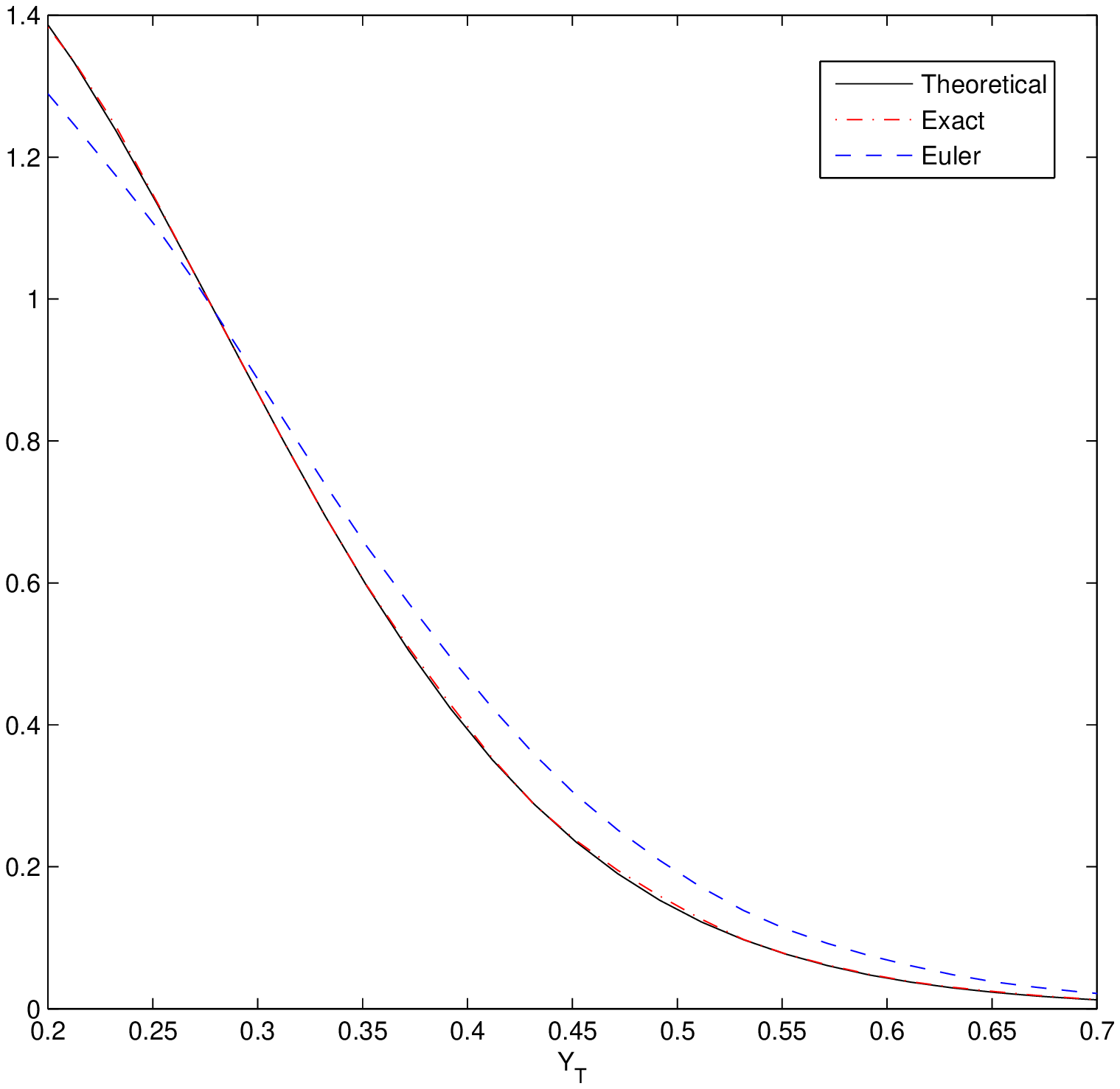}
        \end{subfigure}
        \caption{Density estimates of the log-price :  The upper left panel shows the overall shapes of density estimates, and
            the other three panels magnify different parts of the upper left one. }
        \label{figure:density}
\end{figure}


The errors in density estimates result from the discretization bias of the Euler scheme. Even worse, there is no appropriate way to measure the bias errors.
In contrast, the error of exact simulations comes mostly from the variance, and it can be easily measured by the sample variance. To illustrate such
a difference, we present European call option price estimates obtained by two simulation methods. The model parameters are the same as above except
$\delta = 3.2$. We consider an {\it at-the-money} call option, i.e., the strike price $K$ is set to equal to $S_0 = e^{Y_0} = 1$.
We computed the theoretical price of the option by applying the Carr and Madan approach \cite{CM99} to the Laplace transform formula (\ref{eq:Laplace}).
Table \ref{table:MCDGT1} shows the Monte-Carlo estimates of the option price.
The Euler method gives unreliable results : the theoretical price lies in only two (out of eight) $95\%$ confidence intervals built by the Euler method.
Furthermore, the computation times spent by the Euler method are much longer than those for the exact simulation method.

\begin{table}
    \centering
    \begin{tabular}{cccccc}
        \hline
        \multirow{2}{*}{Methods} & No. of & No. of & \multirow{2}{*}{MC estimates} & \multirow{2}{*}{std. errors} & Time\\
        & time steps & simulation runs &  &  & (sec)\\
        \hline
        \multirow{4}{*}{Exact} & \multirow{4}{*}{N/A} & $50000$ & $0.192415$ & $0.1430 \times 10^{-2}$ & $12.14$\\
        & & $100000$ & $0.191451$ & $0.1009 \times 10^{-2}$ & $24.28$\\
        & & $500000$ & $0.192143$ & $0.4535 \times 10^{-3}$ & $121.4$\\
        & & $1000000$ & $0.191513$ & $0.3202 \times 10^{-3}$ & $242.8$\\
        \hline
        \multirow{8}{*}{Euler}& \multirow{4}{*}{50} & $50000$ & $0.195202^*$ & $0.1456\times 10^{-2}$ & $179.4$\\
        & & $100000$ & $0.194829^*$ & $0.1034\times10^{-2}$ & $358.7$\\
        & & $500000$ & $0.193827^*$ & $0.4603\times10^{-3}$ & $1793.5$\\
        & & $1000000$ & $0.194197^*$ & $0.3259\times10^{-3}$ & $3587.0$\\
        \cline{2-6}
        & \multirow{4}{*}{100} & $50000$ & $0.194016$ & $0.1461\times10^{-2}$ & $358.4$\\
        & & $100000$ & $0.193264$ & $0.1027\times10^{-2}$ & $716.8$\\
        & & $500000$ & $0.193008^*$ & $0.4570\times10^{-3}$ & $3584.2$\\
        & & $1000000$ & $0.193073^*$ & $0.3234\times10^{-3}$ & $7168.5$\\
        \hline
    \end{tabular}
    \caption{The Monte-Carlo estimates of the call option price with $\delta = 3.2$: The theoretical option price is $0.191575$.
        The asterisked numbers are those for which the theoretical price lies outside of the $95\%$ confidence interval.}
    \label{table:MCDGT1}
\end{table}

The accuracy of the Euler method gets worse as $\delta$ decreases. For a small $\delta$, the discretized
Wishart process crosses more frequently the boundary of the cone of symmetric positive definite matrices. At each time it passes the boundary,
its negative part is truncated to make it positive semidefinite, and such truncations might cause serious bias error. To illustrate such case,
we set $\delta = 1.1$, and all other parameters are set as the same as above. Table \ref{table:MCDGT2} reveals that the Euler discretization method
hardly gives reliable estimates for $\delta = 1.1$. The theoretical price never lies in the $95\%$ confidence intervals built by the Euler method,
but it always lies in the $95\%$ confidence intervals from the exact simulation.

\begin{table}
    \centering
    \begin{tabular}{cccccc}
        \hline
        \multirow{2}{*}{Methods} & No. of & No. of & \multirow{2}{*}{MC estimates} & \multirow{2}{*}{std. errors} & Time\\
        & time steps & simulation runs &  &  & (sec)\\
        \hline
        \multirow{4}{*}{Exact} & \multirow{4}{*}{N/A} & $50000$ & $0.114129$ & $0.7381 \times 10^{-3}$ & $15.34$\\
        & & $100000$ & $0.113630$ & $0.5215 \times 10^{-3}$ & $30.68$\\
        & & $500000$ & $0.112916$ & $0.2320 \times 10^{-3}$ & $153.40$\\
        & & $1000000$ & $0.113085$ & $0.1642 \times 10^{-3}$ & $306.80$\\
        \hline
        \multirow{8}{*}{Euler}& \multirow{4}{*}{50} & $50000$ & $0.117013^*$ & $0.7654\times 10^{-3}$ & $180.33$\\
        & & $100000$ & $0.117020^*$ & $0.5441\times10^{-3}$ & $360.67$\\
        & & $500000$ & $0.116698^*$ & $0.2428\times10^{-3}$ & $1803.34$\\
        & & $1000000$ & $0.116940^*$ & $0.1718\times10^{-3}$ & $3606.67$\\
        \cline{2-6}
        & \multirow{4}{*}{100} & $50000$ & $0.116940^*$ & $0.7541\times10^{-3}$ & $360.54$\\
        & & $100000$ & $0.115963^*$ & $0.5321\times10^{-3}$ & $721.09$\\
        & & $500000$ & $0.115333^*$ & $0.2383\times10^{-3}$ & $3605.45$\\
        & & $1000000$ & $0.115363^*$ & $0.1684\times10^{-3}$ & $7210.90$\\
        \hline
    \end{tabular}
    \caption{The Monte-Carlo estimates of the call option price with $\delta = 1.1$: The theoretical option price is $0.113000$.
        The asterisked numbers are those for which the theoretical price lies outside of the $95\%$ confidence interval.}
    \label{table:MCDGT2}
\end{table}

\subsection{Comparison with Broadie and Kaya Method}
To apply our method to the Heston model, we used the explicit formula (\ref{eq:condChfH}) instead of (\ref{eq:condChf}).
With (\ref{eq:condChfH}), we do not need to solve the ordinary differential equations, and we can use the built-in MATLAB function {\sf besseli}
\footnote{The MATLAB funciton {\sf besseli} computes the modified Bessel function of the first kind $I_{\nu}(\cdot)$ using the algorithm of Amos
\cite{Am86}.} instead of the algorithm of Koev and Edelman \cite{KE06} for the computation of the hypergeometric function of matrix arguments.

We pointed out in Section \ref{sec:issue} that a careful choice of the grid size $h$ and the truncation number $N$ is necessary in our method.
The same is true for any Fourier inversion based simulation method, e.g., Broadie and Kaya's method \cite{BK06}.
In the paper of Broadie and Kaya, they did not give any guideline on how to choose $h$ and $N$, and they argued that appropriate $h$ and $N$ can be
found by trial and error. But it is difficult to find appropriate $h$ and $N$, and it will be illustrated by numerical experiments.

This section is designed to show how the various choices of $h$ and $N$ in (\ref{eq:BKdist}) affect the accuracy and performance of the Broadie and Kaya's method,
and to compare performance of their method with ours. For the purpose, we generate $10^6$ samples of the Heston model using our method, and we also
generate the same number of samples using Broadie and Kaya's method for various combinations of $h$ and $N$. We used them to obtain Monte-Carlo estimates of the European
call option price. The model parameter is set to
\begin{eqnarray*}
    x &=& 0.010201, \hspace{0.5cm} y = \log(100), \hspace{0.5cm} \kappa = 6.21, \hspace{0.5cm} \theta = 0.019,\\
    \sigma &=& 0.61, \hspace{0.5cm} \rho = -0.7, \hspace{0.5cm} r = 0.0319, \hspace{0.5cm} T = 1.0, \hspace{0.5cm} K = 100.
\end{eqnarray*}
This set of parameters is taken from Duffie et al. \cite{DPS00}: it was calibrated to the option prices for S\&P 500 on November 2, 1993, and
used in Broadie and Kaya \cite{BK06} as well. Table \ref{tab:heston1} shows the simulation results.\footnote{In our method, the $h$ and $N$ are
automatically chosen by (\ref{eq:choice})} From the table, the pair of
$h = 32$ and $N = 25$ seems to be the best choice of $h$ and $N$ among outputs of the experiments.
But, even with that choice of $h$ and $N$, the Broadie and Kaya's method is slower than our method.

\begin{table}
    \centering
    \begin{tabular}{ccccccc}
        \hline
        Methods & $h$ & $N$ & $N h$ &  MC estimates & std. errors & Time(sec)\\
        \hline
        KK &  N/A & N/A & N/A & $6.8125$ & $7.4302\times10^{-3}$ & $59.77$\\
        \hline
        \multirow{12}{*}{BK} & \multirow{3}{*}{$8$} & $100$ & $800$ & $6.8276^*$&$7.4377\times 10^{-3} $&$280.4$\\
        & & $200$ & $1600$ & $6.8059$&$7.4206\times 10^{-3} $&$591.1$\\
        & & $400$ & $3200$ &$6.7993$&$7.4208\times 10^{-3} $&$1172.0$\\
        \cline{2-7}
        &\multirow{3}{*}{$16$} & $50$ & $800$ &$6.8162$ & $7.4299\times 10^{-3}$&$123.4$\\
        & & $100$ & $1600$ & $6.8025$ & $7.4220\times 10^{-3}$ & $297.0$\\
        & & $200$ & $3200$ & $6.8211^*$ & $7.4257\times 10^{-3}$ & $584.8$\\
        \cline{2-7}
        &\multirow{3}{*}{$32$}& $25$ & $800$ & $ 6.8098$ &$ 7.4302\times 10^{-3}$&$68.0$\\
        & & $50$ & $1600$ & $6.8011$&$ 7.4144\times 10^{-3}$&$132.6$\\
        & & $100$& $3200$ & $6.8191$&$ 7.4286\times 10^{-3}$&$293.6$\\
        \cline{2-7}
        &\multirow{3}{*}{$64$}& $13$ & $832$ & $6.8346^*$&$ 7.4252\times 10^{-3}$&$41.16$\\
        & & $25$ & $1600$ & $6.8222^*$&$7.4173\times 10^{-3} $&$72.52$\\
        & & $50$ & $3200$ & $6.8184$&$7.4137\times 10^{-3} $&$130.8$\\
        \hline
    \end{tabular}
    \caption{The Monte-Carlo estimates of the call option price with $T = 1.0$: The theoretical option price is $6.8061$.
    The asterisker numbers are those for which the theoretical price lies outside of the $95\%$ confidence interval.
    KK refers to our method and BK refers to the method of Broadie and Kaya. }
    \label{tab:heston1}
\end{table}

Through this experiment, we found that the Broadie-Kaya method requires many trials to find the optimal choice of $h$ and $N$. It is even worse
that the optimal choice of $h$ and $N$ is sensitive to the change of the model parameters. To demonstrate the sensitivity of their method,
we give another experimental results with the same parameter except $T = 0.25$. Table \ref{tab:heston2} shows that the true option price
never lies in the confidence intervals of their methods for the same choices of $h$ and $N$. In contrast, our automatic choice of
$h$ and $N$ according to (\ref{eq:choice}) makes our method adaptive to the parameter changes.

\begin{table}
    \centering
    \begin{tabular}{ccccccc}
        \hline
        Methods & $h$ & $N$ & $N h$ &  MC estimates & std. errors & Time(sec)\\
        \hline
        KK &  N/A & N/A & N/A & $2.6720$ & $2.9214\times10^{-3}$ & $78.18$\\
        \hline
        \multirow{12}{*}{BK} & \multirow{3}{*}{$8$} & $100$ & $800$ & $2.7883^*$&$2.8438\times 10^{-3}$ & $279.5$\\
        & & $200$ & $1600$ & $2.7597^*$ &$2.8631\times 10^{-3}$ & $608.0$\\
        & & $400$ & $3200$ & $2.7463^*$ &$2.8705\times 10^{-3}$ & $1290.2$\\
        \cline{2-7}
        &\multirow{3}{*}{$16$} & $50$ & $800$ &$2.7853^*$ & $2.8435\times 10^{-3}$&$124.0$\\
        & & $100$ & $1600$ & $2.7587^*$ & $2.8664\times 10^{-3}$ & $306.6$\\
        & & $200$ & $3200$ & $2.7483^*$ & $2.8796\times 10^{-3}$ & $647.7$\\
        \cline{2-7}
        &\multirow{3}{*}{$32$}& $25$ & $800$ & $2.7802^*$ & $2.8408\times 10^{-3}$ & $67.3$\\
        & & $50$ & $1600$ & $2.7610^*$ & $2.8678\times 10^{-3}$ & $138.7$\\
        & & $100$& $3200$ & $2.7488^*$ & $2.8798\times 10^{-3}$ & $327.9$\\
        \cline{2-7}
        &\multirow{3}{*}{$64$}& $13$ & $832$ & $2.7845^*$ & $2.8433\times 10^{-3}$&$40.63$\\
        & & $25$ & $1600$ & $2.7651^*$ &$2.8695\times 10^{-3}$ & $75.19$\\
        & & $50$ & $3200$ & $2.7445^*$ &$2.8736\times 10^{-3}$ & $150.9$\\
        \hline
    \end{tabular}
    \caption{The Monte-Carlo estimates of the call option price with $T = 0.25$: The theoretical option price is $2.6709$.
    The asterisked numbers are those for which the theoretical price lies outside of the $95\%$ confidence interval.
    KK refers to our method and BK refers to the method of Broadie and Kaya. }
    \label{tab:heston2}
\end{table}

\section{Concluding Remarks} \label{sec:con}
We proposed a method for the exact simulation of the asset price and volatility factor of Wishart multidimensional stochastic volatility model(WMSV).
Our simulation method is inspired by the Fourier inversion techniques of Broadie and Kaya \cite{BK06}, and based on the analysis of the conditional
Laplace transform of the asset price given a volatility level. The proposed simulation method can be used to generate unbiased price estimators for
path-independent or mildly path-dependent options such as Bermudan options and American call options with discrete dividends.

We illustrated the accuracy and speed of the exact simulation method by numerical experiments. The numerical results confirmed that the exact
simulation method gives accurate simulation results and it is faster than the crude Euler discretization method.
The numerical comparison on the Heston model revealed that our method is more adaptive than the original exact simulation method
of Broadie and Kaya.

In the paper of Broadie and Kaya \cite{BK06}, the authors extended the exact simulation method for the Heston model to other affine jump diffusion models.
With the same approach as theirs, the exact simulation method for the WMSV model can be extended to single asset matrix affine jump
diffusion models \cite{LT08} provided that the jumps in the model has the constant intensity and the jump size can be sampled exactly.
These extensions are straightforward and explained well in Broadie and Kaya \cite{BK06}.

\appendix

\section{Appendix}

\subsection{Noncentral Wishart Distributions \& Special Functions} \label{sec:def}
This section is intended to recall the definitions of noncentral Wishart distributions and some multivariate special functions.
For a detailed discussion on multivariate distributions and special functions, refer to Muirhead \cite{Mu82}.

The probability density function of the noncentral Wishart distribution involves two special functions: the multivariate gamma function
and the hypergeometric function of matrix arguments. We start with the multivariate gamma function, which is defined in terms of
an integral over the cone $S_d^{++}$ of symmetric positive definite $d \times d$ matrices.

\begin{defn}
The multivariate gamma function, denoted by $\Gamma_d(a)$, is defined to be
\begin{equation*}
    \Gamma_d(a) = \int_{S_d^{++}} \exp(-\tr[y])(\dete[y])^{a - (d+1)/2} (dy),
        \hspace{0.5cm} \text{ for } \hspace{0.5cm} \re(a) > \frac{1}{2}(d - 1).
\end{equation*}
\end{defn}

Note that when $d = 1$, the multivariate gamma function becomes the usual gamma function, $\Gamma_1(a) = \Gamma(a)$.

In order to give the definition of the hypergeometric function of matrix arguments, we need to introduce the zonal polynomials,
which is defined in terms of partitions of positive integers. Let $k$ be a positive integer. A partition $\iota$ of
$k$ is written as $\iota = (k_1, k_2, \cdots)$, where $\sum_j k_j = k$ and $k_1 \ge k_2 \ge \cdots \ge 0$. We order the partitions
lexicographically: let $\iota = (k_1, k_2, \cdots)$ and $\tilde{\iota} = (l_1, l_2, \cdots)$ be two partitions, we write $\iota > \hat{\iota}$ if
$k_j > l_j$ for the first index $j$ at which two parts become unequal. In case $\iota = (k_1, \cdots, k_d)$ and $\tilde{\iota} = (l_1, \cdots, l_d)$ are two
partitions with $\iota > \tilde{\iota}$, we say the monomial $\alpha_1^{k_1}\cdots\alpha_d^{k_d}$ is of higher weight than the monomial
$\alpha_1^{l_1}\cdots\alpha_d^{l_d}$.

\begin{defn}
Let $y$ be a $d \times d$ complex symmetric matrix with eigenvalues $\alpha_1, \cdots, \alpha_d$ and let $\iota$ be a partition of $k$ into not more than
$d$ parts. The zonal polynomial of $y$ corresponding to $\iota = (k_1, \cdots, k_d)$, denoted by $C_{\iota}(y)$, is a symmetric, homogeneous polynomial
of degree $k$ in the eigenvalues $\alpha_1, \cdots, \alpha_d$ such that
\begin{enumerate}[(i)]
    \item The term of highest weight in $C_{\iota}(y)$ is $\alpha_1^{k_1}\cdots\alpha_d^{k_d}$, i.e.,
        \begin{equation*}
            C_{\iota}(y) = c_{\iota}\alpha_1^{k_1}\cdots\alpha_d^{k_d} + \text{ terms of lower weight},
        \end{equation*}
        where $c_{\iota}$ is a constant.
    \item $C_{\iota}(y)$ is an eigenfunction of the differential operator $\Delta_{y}$ given by
        \begin{equation*}
            \Delta_{y} = \sum_{j = 1}^d \alpha_j^2 \frac{\partial^2}{\partial \alpha_j^2}
                + \sum_{j = 1}^d \sum_{\substack{ i = 1 \\ i \neq j}}^d \frac{\alpha_j^2}{\alpha_j - \alpha_i} \frac{\partial}{\partial \alpha_j}.
        \end{equation*}
        In other words, $C_{\iota}(y)$ satisfies $\Delta_{y} C_{\iota}(y) = \lambda C_{\iota}(y)$ for some $\lambda$.
    \item As $\iota$ varies over all partitions of $k$, the zonal polynomials have unit coefficients in the expansion of $(\tr[y])^k$, i.e.,
        \begin{equation*}
            \sum_{|\iota| = k} C_{\iota} (y) = (\tr[y])^k = (\alpha_1 + \cdots + \alpha_d)^k,
        \end{equation*}
        where $|\iota|$ denotes the sum of parts of $\iota$, i.e. $|\iota| = k_1 + \cdots + k_d$.
\end{enumerate}
\end{defn}

\begin{defn}
The hypergeometric functions of matrix argument are given by
\begin{equation*}
    {}_p F_q (a_1, \cdots, a_p; b_1, \cdots, b_q; y)
        = \sum_{k=0}^{\infty} \sum_{|\iota| = k} \frac{(a_1)_{\iota} \cdots (a_p)_{\iota}} {(b_1)_{\iota} \cdots (b_p)_{\iota}}
        \frac{C_{\iota}(y)}{k!},
\end{equation*}
where the generalized hypergeometric coefficient $(a)_{\iota}$ for a partition $\iota = (k_1, \cdots, k_d)$ is given by
\begin{equation*}
    (a)_{\iota} = \prod_{j = 1}^d \Big(a - \frac{1}{2}(j - 1)\Big)_{k_j},
\end{equation*}
and $(a)_k = a(a+1)\cdots(a + k - 1)$, $(a)_0 = 1$.
Here, the argument of the function is a complex symmetric $d \times d$ matrix and the parameters $a_j, b_j$ are arbitrary complex numbers. No denominator
parameter $b_j$ is allowed to be zero or an integer or half-integer $\le \frac{1}{2} (d - 1)$, otherwise some of the denominators in the series will
vanish.
\end{defn}

Note that in case $d = 1$, there is only one partition of $k$ into not more than $d$, namely $(k)$. Therefore, the hypergeometric function of matrix arguments
becomes to the usual hypergeometric function of real variables.

Likewise the noncentral chi-square distributions, the noncentral Wishart distribution with integer degrees of freedom is defined by multiplying normal
random vectors with themselves. Consider $\delta (\in \mathbb{N})$ independent random vectors of $\reals^d$, denoted by $Z_1, \cdots, Z_{\delta}$, with
multivariate normal distributions with means $\mu_1, \cdots, \mu_{\delta}$ and the common nonsingular covariance matrix $C$. The distribution of the
random matrix
\begin{equation*}
    W = \sum_{k = 1}^{\delta} Z_k Z_k^{\top}
\end{equation*}
is called the noncentral Wishart distribution with $\delta$ degrees of freedom, covariance matrix $C$, and matrix of noncentrality parameter
$\Omega = C^{-1} \sum_{k = 1}^{\delta} \mu_k \mu_k^{\top}$. If $\delta \ge d$, then $W$ has a probability density function which is of the form:
\begin{equation}\label{eq:wisden}
    \frac{(\dete[y])^{(\delta - d - 1)/2}}
        {2^{d\delta/2}\Gamma_d(\textstyle \frac{1}{2} \delta) (\dete[C])^{\delta/2}}
        \exp\Big\{\textstyle - \frac{1}{2} \tr\big[ C^{-1}y + \Omega \big]\Big\}
        {}_0 F_1\Big(\textstyle \frac{1}{2}\delta ;
        \frac{1}{4} \Omega C^{-1} y\Big), \hspace{0.5cm} y \in S_d^{++}.
\end{equation}
The function (\ref{eq:wisden}) is still a density function when $\delta$ is any real number greater than $d - 1$. This observation make it possible
to extend the notion of the noncentral Wishart distribution to non-integer degrees of freedom.

\begin{defn}
Suppose that $\delta > d - 1$, $C \in S_d^{++}$, and $\Omega$ is a $d \times d$ matrix such that $C\Omega$ is symmetric positive semidefinite.
If a $d \times d$ symmetric positive definite random matrix $W$ has the probability density function (\ref{eq:wisden}), then $W$ is said to
have the noncentral Wishart distribution with $\delta$ degrees of freedom, covariance matrix $C$, and matrix of noncentrality parameter $\Omega$.
We will say $W$ has $\mathcal{W}_{d}(\delta, C, \Omega)$ distribution.
\end{defn}

\begin{prop}[Theorem 3.5.3 in Gupta and Nagar \cite{GN00}]\label{prop:wislaplace}
Suppose $W$ has noncentral Wishart distribution $\mathcal{W}_{d}(\delta, C, \Omega)$. Then its Laplace transform is given by
\begin{equation*}
\pE\big[ e^{ - \tr(\vartheta W)}\big] = \dete[I_d + 2\vartheta C]^{-\delta/2}
    \exp\Big\{- \tr\big[\vartheta(I_d + 2C\vartheta)^{-1} C \Omega\big]\Big\},
\end{equation*}
for any symmetric positive semidefinite matrix $\vartheta$.
\end{prop}

\subsection{Wishart Processes with Time-varying Linear Drift}\label{sec:wisharttime}
In this section, we slightly extend the notion of Wishart processes in order to compute the conditional Laplace transform of log-price given volatility
level.

A symmetric positive semidefinite matrix valued stochastic process $X$ is called a Wishart process with time-varying linear drift if it is a
weak solution of the following stochastic differential equation
\begin{equation}\label{eq:wisharttime}
    dX_t = (\delta\Sigma^{\top}\Sigma + H(t)X_t + X_t H(t)\tran ) dt + \sqrt{X_t}dW_t \Sigma + \Sigma\tran dW_t\tran \sqrt{X_t},
        \text{ with } X_0 = x,
\end{equation}
where $\delta \ge d-1$, $\Sigma$ is a $d\times d$ matrix, $x$ is a symmetric positive semidefinite matrix, $H(\cdot)$ is
a $d \times d$ matrix valued continuous function, and $W$ is a standard $d \times d$ matrix Brownian motion.
Wishart process with time-varying linear drift has noncentral Wishart marginal distributions.

\begin{prop}\label{prop:wisharttimelaplace}
Let $X$ be a Wishart process with time-varying linear drift which solves (\ref{eq:wisharttime}). Then $X_T$ has
noncentral Wishart distribution $\mathcal{W}_{d}$$(\delta$, $V(0)$, $V(0)^{-1}$ $\Psi(0)^{\top}$$x$$\Psi(0))$, where $V(t)$ and $\Psi(t)$
are solutions of following system of ordinary differential equations
\begin{equation*}
    \left\{
        \begin{array}{rcl}
            \frac{d}{dt} \Psi(t) &=& - H(t)^{\top} \Psi(t),\\
            \frac{d}{dt} V(t) &=& - \Psi(t)^{\top} \Sigma^{\top}\Sigma \Psi(t),\\
        \end{array}
    \right.
\end{equation*}
with the terminal value $\Psi(T) = I_d$ and $V(T) = 0$.
\begin{proof}
Using standard argument(e.g., see Appendix B of Gourieroux and Sufana \cite{GS10}), one may prove that
for a symmetric positive semidefinite matrix $\vartheta$,
\begin{equation*}
\pE_x\big[ e^{ - \tr(\vartheta X_T)}\big] = e^{ - \hat{\phi}(0, \vartheta) - \tr(\hat{\psi}(0, \vartheta)x)},
\end{equation*}
where $\hat{\phi}$ and $\hat{\psi}$ are the solution of following equations
\begin{equation*}
    \left\{
        \begin{array}{ll}
            \partial _t \hat{\psi}(t, \vartheta) =& 2 \hat{\psi}(t, \vartheta) \Sigma^{\top} \Sigma \hat{\psi}(t, \vartheta)
                - H(t)^{\top}\hat{\psi}(t,\vartheta) - \hat{\psi}(t,\vartheta)H(t),\\
            \partial_t \hat{\phi}(t, \vartheta) =& -\delta \tr[\hat{\psi}(t,\vartheta)\Sigma^{\top}\Sigma],
        \end{array}
    \right.
\end{equation*}
with terminal values $\hat{\phi}(T,\vartheta) = 0$ and $\hat{\psi}(T, \vartheta) = \vartheta$. Using differentiation rules
$\frac{d}{dt} A(t)^{-1} = - A(t)^{-1} \frac{d}{dt} A(t) A(t)^{-1}$ and $\frac{d}{dt} \ln \dete[A(t)] = \tr[A(t)^{-1} \frac{d}{dt} A(t)]$,
one may check that
\begin{equation*}
    \left\{
        \begin{array}{ll}
            \hat{\psi}(t, \vartheta) =& \Psi(t) \vartheta(I_d + 2 V(t) \vartheta)^{-1} \Psi(t)^{\top},\\
            \hat{\phi}(t, \vartheta) =& \frac{\delta}{2} \ln \dete[I_d + 2 \vartheta V(t)],
        \end{array}
    \right.
\end{equation*}
solves the above system of differential equations. Therefore,
\begin{equation*}
\pE_x\big[ e^{ - \tr(\vartheta X_T)}\big] = \dete[I_d + 2\vartheta V(0)]^{-\delta/2}
    \exp\Big\{- \tr\big[\vartheta(I_d + 2V(0)\vartheta)^{-1} \Psi(0)^{\top} x \Psi(0)\big]\Big\}.
\end{equation*}
Since Laplace transform uniquely characterizes a distribution, $X_T$ has noncentral Wishart distribution $\mathcal{W}_{d}$$(\delta$, $V(0)$,
$V(0)^{-1}$ $\Psi(0)^{\top}$$x$$\Psi(0))$ by Proposition \ref{prop:wislaplace}.
\end{proof}
\end{prop}

\subsection{Details of calculation of (\ref{eq:condChfH})}\label{sec:app2}

In this section, we provide details of the calculation of the conditional characteristic function (\ref{eq:condChfH}).
The first equation in the system (\ref{eq:systemheston}) is the classical Riccati equation, and
its closed-form solution is well-known (e.g., see Section 10.7.2 of Filipovi\'{c} \cite{Fi09}). In particular,
\begin{eqnarray*}
    \psi(t,u) &=& - \frac{ u(u+1) (e^{\eta(u)(T-t)} - 1)}
        { \eta(u) (e^{\eta(u)(T-t)} + 1) + (\kappa + u \sigma \rho)(e^{\eta(u)(T-t)} - 1)},\\
    \int_t^T \psi(s,u) d s &=& -\frac{2}{\sigma^2}
        \log \left( \frac{ 2 \eta(u) e^{\frac{1}{2} (\eta(u) + \kappa + u \sigma \rho)(T-t)}}
        {\eta(u) (e^{\eta(u) (T-t)} + 1) + (\kappa + u \sigma \rho)(e^{\eta(u) (T-t)} - 1)} \right),
\end{eqnarray*}
where $\eta(u) = \sqrt{(\kappa + u \sigma \rho)^2 - \sigma^2 u(u+1)}$. It follows that
\begin{equation*}
    e^{-\phi(t,u)} = \left( \frac{ 2 \eta(u) e^{\frac{1}{2} (\eta(u) + \kappa + u \sigma \rho)(T-t)}}
        {\eta(u) (e^{\eta(u) (T-t)} + 1) + (\kappa + u \sigma \rho)(e^{\eta(u) (T-t)} - 1)} \right)^{\delta/2} e^{- u r (T-t)}.
\end{equation*}
The solution of the third linear equation is
\begin{eqnarray}
    \Psi(t,u) &=& \exp\big\{ \textstyle - \frac{1}{2} \int_t^T (\kappa + u \sigma \rho + \sigma^2 \psi(s,u)) d s\big\}\nonumber\\
    &=& e^{- \frac{1}{2} (\kappa + u \sigma \rho)(T-t)}
        \exp\big\{\textstyle - \frac{\sigma^2}{2} \int_t^T \psi(s,u) d s \big\}\nonumber\\
    &=& e^{- \frac{1}{2} (\kappa + u \sigma \rho)(T-t)}
        \frac{ 2 \eta(u) e^{\frac{1}{2} (\eta(u) + \kappa + u \sigma \rho)(T-t)}}
        {\eta(u) (e^{\eta(u) (T-t)} + 1) + (\kappa + u \sigma \rho)(e^{\eta(u) (T-t)} - 1)}\nonumber\\
    &=& \frac{ 2 \eta(u) e^{\frac{1}{2} \eta(u)(T-t)}}
        {\eta(u) (e^{\eta(u) (T-t)} + 1) + (\kappa + u \sigma \rho)(e^{\eta(u) (T-t)} - 1)}\label{eq:PsiH}
\end{eqnarray}
A direct integration shows that
\begin{equation}\label{eq:VH}
    V(t,u) = \frac{1}{2} \frac{\sigma^2 (e^{\eta(u) (T-t)} - 1)}{\eta(u)(e^{\eta(u) (T-t)} + 1) + (\kappa + u \sigma \rho)(e^{\eta(u) (T-t)} - 1)}.
\end{equation}
We divide (\ref{eq:PsiH}) by (\ref{eq:VH}) to have
\begin{equation*}
    \frac{\Psi(0,u)}{V(0,u)} = \frac{4 \eta(u) e^{0.5 \eta(u) T}}{\sigma^2(e^{\eta(u)T} - 1)}
        = \frac{4 \eta(u) e^{ - 0.5 \eta(u) T}}{\sigma^2(1 - e^{-\eta(u)T})}.
\end{equation*}
In particular,
\begin{equation*}
    \frac{\Psi(0,0)}{V(0,0)} = \frac{4 \kappa e^{-0.5\kappa T}}{\sigma^2(1 - e^{-\kappa T})}.
\end{equation*}
Observe that the relation between $\phi(t,u)$ and $\Psi(t,u)$:
\begin{equation*}
    e^{-\phi(t,u)} = (\Psi(t,u))^{\delta/2} \exp\Big\{\textstyle \frac{\delta}{4}(\kappa + u \sigma \rho)(T-t) - u r (T-t)\Big\}.
\end{equation*}
Remind the relationship between the hypergeometric functions and the modified Bessel functions:
\begin{equation*}
    {}_0 F_1\Big(\nu + 1; {\textstyle \frac{1}{4} x^2}\Big) = (x/2)^{-\nu} \Gamma(\nu + 1) I_{\nu}(x).
\end{equation*}
Using the identities above, we found that
\begin{eqnarray*}
    \lefteqn{\Big(\frac{V(0,0)}{V(0,u)}\Big)^{\delta/2} \exp\big\{-\phi(0,u)\big\} \frac{ {}_0 F_1\big({\textstyle \frac{1}{2} \delta};
        {\textstyle \frac{1}{4} \big(\frac{\Psi(0,u)}{V(0,u)}\sqrt{x x_{\scriptscriptstyle T}}\big)^2}\big)}
        { {}_0 F_1\big({\textstyle \frac{1}{2} \delta};
        {\textstyle \frac{1}{4} \big(\frac{\Psi(0,0)} {V(0,0)}\sqrt{x x_{\scriptscriptstyle T}}\big)^2}\big)}}\\
    &=& \Big(\frac{V(0,0)}{V(0,u)}\Big)^{\delta/2} \Big(\frac{\Psi(0,u)}{V(0,u)}\Big)^{-\delta/2+1} \Big(\frac{\Psi(0,0)}{V(0,0)}\Big)^{\delta/2-1}\\
    & & \times \big(\Psi(0,u)\big)^{\delta/2}\exp\big\{{\textstyle \frac{\delta}{4}(\kappa + u \sigma \rho)T - u r T}\big\}
        \frac{I_{0.5\delta-1}(\sqrt{x x_{\scriptscriptstyle T}}\frac{\Psi(0,u)}{V(0,u)})}
        {I_{0.5\delta-1}(\sqrt{x x_{\scriptscriptstyle T}}\frac{\Psi(0,0)}{V(0,0)})}\\
    &=& V(0,0) \frac{\Psi(0,u)}{V(0,u)} \exp\big\{{\textstyle \frac{1}{2} \kappa T - u(r - \frac{\kappa \theta \rho}{\sigma})T}\big\}
        \frac{I_{0.5\delta-1}(\sqrt{x x_{\scriptscriptstyle T}}\frac{\Psi(0,u)}{V(0,u)})}
        {I_{0.5\delta-1}(\sqrt{x x_{\scriptscriptstyle T}}\frac{\Psi(0,0)}{V(0,0)})}\\
    &=& \frac{\eta(u)(1 - e^{-\kappa T})}{\kappa(1 - e^{-\eta(u) T})}
        \exp\big\{{\textstyle - u (r - \frac{\kappa\theta\rho}{\sigma})T - \frac{1}{2}(\eta(u) - \kappa)T}\big\}
        \frac{I_{0.5\delta - 1}\big[ \sqrt{x x_{\scriptscriptstyle T}}
        \frac{4\eta(u)e^{-0.5\eta(u)T}}{ \sigma^2(1-e^{-\eta(u)T}) }\big] }
        {I_{0.5\delta - 1}\big[ \sqrt{x x_{\scriptscriptstyle T}} \frac{4\kappa e^{-0.5\kappa T}}{ \sigma^2(1-e^{-\kappa T}) } \big] }.
\end{eqnarray*}
Recall that $\eta(u)^2 = (\kappa + u \sigma \rho)^2 - \sigma^2 u(u+1)$. Using this identity, we have
\begin{eqnarray*}
    \lefteqn{2\psi(0,u) + \frac{\Psi(0,u)^2}{V(0,u)} = \frac{ -2 u(u+1) (e^{\eta(u) T} - 1)}
        { \eta(u) (e^{\eta(u) T} + 1) + (\kappa + u \sigma \rho)(e^{\eta(u) T} - 1)}}\\
    & & +\frac{4 \eta(u) e^{ - 0.5 \eta(u) T}}{\sigma^2(1 - e^{-\eta(u)T})}
        \frac{ 2 \eta(u) e^{0.5 \eta(u) T}}
        {\eta(u) (e^{\eta(u) T} + 1) + (\kappa + u \sigma \rho)(e^{\eta(u) T} - 1)}\\
    &=& \frac{ -2 \sigma^2 u(u+1) (e^{\eta(u) T} - 1)(1 - e^{-\eta(u)T}) + 8 \eta(u)^2}
        {\sigma^2(1 - e^{-\eta(u)T})(\eta(u) (e^{\eta(u) T} + 1) + (\kappa + u \sigma \rho)(e^{\eta(u) T} - 1))}\\
    &=& \frac{ -2 \sigma^2 u(u+1) (1 - e^{-\eta(u)T})^2 + 8 \eta(u)^2 e^{-\eta(u)T}}
        {\sigma^2(1 - e^{-\eta(u)T})(\eta(u) (1 + e^{-\eta(u)T}) + (\kappa + u \sigma \rho)(1 - e^{-\eta(u)T}))}\\
    &=& \frac{ -2 \sigma^2 u(u+1) - 4 \sigma^2 u(u+1) e^{-\eta(u) T} -2 \sigma^2 u(u+1) e^{-2\eta(u) T}
        + 8 (\kappa + u\sigma \rho)^2 e^{-\eta(u)T}}
        {\sigma^2(1 - e^{-\eta(u)T})(\eta(u) (1 + e^{-\eta(u)T}) + (\kappa + u \sigma \rho)(1 - e^{-\eta(u)T}))}\\
    &=& \frac{ -2 \sigma^2 u(u+1)(1 + e^{-\eta(u) T})^2 + 8 (\kappa + u\sigma \rho)^2 e^{-\eta(u)T}}
        {\sigma^2(1 - e^{-\eta(u)T})(\eta(u) (1 + e^{-\eta(u)T}) + (\kappa + u \sigma \rho)(1 - e^{-\eta(u)T}))}\\
    &=& \frac{2}{\sigma^2}\frac{ (\eta(u)^2 - (\kappa + u \sigma \rho)^2) (1 + e^{-\eta(u) T})^2 + 4 (\kappa + u\sigma \rho)^2 e^{-\eta(u)T}}
        {(1 - e^{-\eta(u)T})(\eta(u) (1 + e^{-\eta(u)T}) + (\kappa + u \sigma \rho)(1 - e^{-\eta(u)T}))}\\
    &=& \frac{2}{\sigma^2}\frac{ \eta(u)^2 (1 + e^{-\eta(u) T})^2 - (\kappa + u\sigma \rho)^2 (1 - e^{-\eta(u)T})^2}
        {(1 - e^{-\eta(u)T})(\eta(u) (1 + e^{-\eta(u)T}) + (\kappa + u \sigma \rho)(1 - e^{-\eta(u)T}))}\\
    &=& \frac{2}{\sigma^2}\frac{ \eta(u) (1 + e^{-\eta(u) T}) - (\kappa + u\sigma \rho) (1 - e^{-\eta(u)T})}
        {1 - e^{-\eta(u)T}}.
\end{eqnarray*}
And
\begin{equation*}
    \frac{1}{V(0,0)} = \frac{2}{\sigma^2}\frac{2\kappa}{1 - e^{-\kappa T}}
        \hspace{0.5cm} \text{ and } \hspace{0.5cm}
        \frac{\Psi(0,0)^2}{V(0,0)} = \frac{2}{\sigma^2}\frac{2\kappa e^{-\kappa T}}{1 - e^{-\kappa T}}.
\end{equation*}
By substituting above quantities into the formula (\ref{eq:condChf}), we found the closed-form expression (\ref{eq:condChfH}) for the conditional
Laplace transform of the Heston model.

\subsection{Proof of Proposition \ref{prop:hyperub}}\label{sec:app1}
We start with a simple lemma.
\begin{lem}\label{lemma:1}
Let $b > \frac{1}{2}(d - 1)$. Suppose $\hat{\iota}(k)$ is the smallest partition among the partitions of $k$ into not more than $d$ parts. Then
\begin{equation*}
    0 < (b)_{\hat{\iota}(k)} \le (b)_{\iota}
\end{equation*}
for all partitions $\iota$ of $k$ into not more than $d$ parts.
\begin{proof}
Fix $k$ and we simplify the notation $\hat{\iota}(k)$ as $\hat{\iota}$.
The factors of the hypergeometric coefficients are of the form
\begin{equation*}
    0 < b  - \frac{1}{2}(d-1) \le b + l - \frac{1}{2}(j-1) \le b + k - 1,
\end{equation*}
for $j = 1, \cdots, d$, and $l = 0, \cdots, k_j - 1$. Thus, we have $(b)_{\hat{\iota}} > 0$.
For a partition $\iota = (k_1,\ldots, k_d)$, define $\alpha(\iota) := \max \{k_i - k_j : 1\le i < j \le d\}$ and
$\beta(\iota) := \min\{j-i: k_i - k_j = \alpha(\iota), i < j\}$. Assume that $\iota$ is not the smallest one.
Then $\alpha(\iota) \ge 2$. Let $\alpha(\iota) = k_{i^*} - k_{j^*}$ and $\beta(\iota) = j^*-i^*$.
We define a new partition $\iota' = (k_1, \ldots, k_{i^*-1}, k_{i^*}-1, k_{i^*+1},\ldots, k_{j^*-1}, k_{j^*}+1, k_{j^*+1}, \ldots, k_d)$.
Then $(b)_{\iota'} = (b)_{\iota} \times \frac{b+k_{j^*} - (j^*+1)/2}{b + k_{i^*} - (i^*+1)/2} < (b)_{\iota}$. For the new partition $\iota'$,
if we have $\alpha(\iota') = \alpha(\iota)$, then $\beta(\iota') = \beta(\iota)+2$. Since $\beta(\cdot) \le d-1$,
we should have a partition with decreased $\alpha$ value after some iterations.
Hence, we obtain a partition $\hat{\iota}$ with $\alpha(\hat{\iota}) \le 1$ in finite steps,
which is the smallest one. Since the new partitions always have smaller hypergeometric coefficient values,
we conclude $(b)_{\hat{\iota}} < (b)_{\iota}$ for any non-smallest partition $\iota$.
\end{proof}
\end{lem}

For reader's better understanding, we give an example which illustrates the idea of the above proof. Consider the case $d = 5$ and $k = 8$.
In this case, the smallest partition is
\begin{equation*}
    \hat{\iota} = \hat{\iota}(8) = (2,2,2,1,1).
\end{equation*}
We start with a partition which is not the smallest, say $\iota = (4,3,1,0,0)$.
From this partition, we successively looking for a new partition with a smaller hypergeometric coefficients until we arrive at the smallest one:
\begin{equation*}
    \begin{array}{ccc}
        \begin{array}{|c|c|c|c|c|}
            \multicolumn{5}{c}{(4,3,1,0,0), \hspace{0.15cm} \alpha(\iota) = 4, \beta(\iota) = 3}\vspace{0.2cm}\\
            1& 2 & 3& 4& 5\\
            b + 3& & & &\\
            b + 2& b + \frac{3}{2}& & &\\
            b + 1& b + \frac{1}{2}&  & &\\
            b & b - \frac{1}{2} & b - 1 & \hspace{0.8cm}& \hspace{0.8cm} \\
        \end{array}
        & \longrightarrow &
        \begin{array}{|c|c|c|c|c|}
            \multicolumn{5}{c}{(3, 3, 1, 1, 0), \hspace{0.15cm} \alpha(\iota) = 3, \beta(\iota) = 3}\vspace{0.2cm}\\
            1& 2 & 3& 4& 5\\
            & & & &\\
            b + 2& b + \frac{3}{2}& & &\\
            b + 1& b + \frac{1}{2}& b & &\\
            b    & b - \frac{1}{2}& b - 1 & b - 2& \hspace{0.8cm} \\
        \end{array}
    \end{array}
\end{equation*}

\begin{equation*}
    \begin{array}{cccc}
        \longrightarrow &
        \begin{array}{|c|c|c|c|c|}
            \multicolumn{5}{c}{(3,2,1,1,1), \hspace{0.15cm} \alpha(\iota) = 2, \beta(\iota) = 2}\vspace{0.2cm}\\
            1& 2 & 3& 4& 5\\
            & & & &\\
            b + 2& & & &\\
            b + 1& b + \frac{1}{2}& & &\\
            b & b - \frac{1}{2} & b - 1 & b - 2& b - 3 \\
        \end{array}
        & \longrightarrow &
        \begin{array}{|c|c|c|c|c|}
            \multicolumn{5}{c}{(2, 2, 2, 1, 1), \hspace{0.15cm} \alpha(\hat{\iota}) = 1, \beta(\hat{\iota}) = 1}\vspace{0.2cm}\\
            1& 2 & 3& 4& 5\\
            & & & &\\
            & & & &\\
            b + 1& b + \frac{1}{2}& b & &\\
            b & b - \frac{1}{2} & b - 1 & b - 2& b - 3 \\
        \end{array}
    \end{array}
\end{equation*}
\noindent At each stage, the values of the hypergeometric coefficients decrease, and the smallest partition has the minimal hypergeometric coefficient.
The next lemma is about comparisons between hypergeometric coefficients of partitions with different sizes.

\begin{lem}\label{lemma:2}
Let $b > \frac{1}{2}(d-1)$ and $k \ge d$. Then
\begin{equation*}
    (b)_{\hat{\iota}(k)} < (b)_{\hat{\iota}(k+1)}.
\end{equation*}
\begin{proof}
This lemma is a simple consequence of the recurrence relation (\ref{eq:hyperco}). We write $\hat{\iota}(k) = (k_1, \cdots, k_d)$. Since $k \ge d$,
we have $k_j \ge 1$ for all $j = 1, \cdots, d$. Therefore,
\begin{equation*}
    1 \le k_j < b + k_j - \frac{1}{2}(d-1) \le b + k_j - \frac{1}{2}(j - 1),
\end{equation*}
for all $j = 1, \cdots, d$. This observation and the recurrence relation (\ref{eq:hyperco}) complete the proof.
\end{proof}
\end{lem}

Now we consider the zonal polynomials. Zonal polynomials are polynomials of eigenvalues of the matrix, and their coefficients are all nonnegative
(see the recurrence relation (14) in page 234 of Muirhead \cite{Mu82}). Therefore,
\begin{equation*}
    |C_{\iota}(\alpha_1, \cdots, \alpha_d)| \le C_{\iota}(|\alpha_1|, \cdots, |\alpha_d|),
\end{equation*}
for all partitions $\iota$, and all complex numbers $\alpha_1, \cdots, \alpha_d$.

\begin{proof}[Proof of Proposition \ref{prop:hyperub}]
The proof is a combination of the inequalities which are established in Lemma \ref{lemma:1} and \ref{lemma:2}. Observe that
\begin{eqnarray*}
    \lefteqn{|{}_0 F_1(b;y) - {}_0^m F_{1}^{}(b;y)|
        = \Big|\sum_{k = m+1}^{\infty} \sum_{|\iota| = k} \frac{1}{(b)_{\iota}} \frac{C_{\iota}(y)}{k!}\Big|
        \le \sum_{k = m+1}^{\infty} \sum_{|\iota| = k} \frac{1}{(b)_{\iota}} \frac{|C_{\iota}(y)|}{k!}}\\
    & &\le \sum_{k = m+1}^{\infty} \frac{1}{k! (b)_{\hat{\iota}(k)}} \sum_{|\iota| = k} C_{\iota}(|\alpha_1|, \cdots, |\alpha_d|)
        \le \frac{1}{(b)_{\hat{\iota}(m+1)}} \sum_{k = m+1}^{\infty} \frac{1}{k!}\sum_{|\iota| = k} C_{\iota}(|\alpha_1|, \cdots, |\alpha_d|).
\end{eqnarray*}
By definition, zonal polynomials are normalized so that
\begin{equation*}
    \sum_{|\iota| = k} C_{\iota}(\beta_1, \cdots, \beta_d) = (\beta_1 + \cdots + \beta_d)^k
\end{equation*}
for all complex numbers $\beta_1, \cdots, \beta_d$ (see Definition 7.2.1 of Muirhead \cite{Mu82}). Therefore, we have
\begin{eqnarray*}
    |{}_0 F_1(b;y) - {}_0^m F_{1}^{}(b;y)|
        &\le& \frac{1}{(b)_{\hat{\iota}(m+1)}} \sum_{k = m+1}^{\infty} \frac{1}{k!}(|\alpha_1| + \cdots + |\alpha_d|)^k\\
    &\le& \frac{1}{(b)_{\hat{\iota}(m+1)}} \sum_{k=0}^{\infty} \frac{1}{(k+m+1)!}(|\alpha_1| + \cdots + |\alpha_d|)^{k+m+1}\\
    &\le& \frac{(|\alpha_1| + \cdots + |\alpha_d|)^{m+1}}{(m+1)!(b)_{\hat{\iota}(m+1)}}
        \sum_{k=0}^{\infty} \frac{1}{k!}(|\alpha_1| + \cdots + |\alpha_d|)^k\\
    &=& \frac{(|\alpha_1| + \cdots + |\alpha_d|)^{m+1}}{(m+1)!(b)_{\hat{\iota}(m+1)}} e^{|\alpha_1| + \cdots + |\alpha_d|}.
\end{eqnarray*}
\end{proof}

\bibliographystyle{abbrv}
\bibliography{reference}

\end{document}